\documentclass[journal]{IEEEtran}
\usepackage{amssymb,amsfonts,amsthm,mathtools,bm}
\usepackage{amsmath}
\usepackage{tikz}
\usetikzlibrary{matrix,decorations.pathreplacing}
\usepackage{cite}
\usepackage{array}
\usepackage[caption=false,font=normalsize,labelfont=sf,textfont=sf]{subfig}
\usepackage{textcomp}
\usepackage{verbatim}
\usepackage{graphicx,stfloats,epstopdf,url}
\usepackage{makecell}
\usepackage{soul, color, xcolor}

\usepackage{algorithm, algorithmic}
\usepackage{booktabs}
\usepackage[utf8]{inputenc}
 \usepackage{balance}
 \usepackage{enumitem}
\usepackage{tabularx}
\usepackage{multirow}
\usepackage{relsize}

\newtheorem{lemma}{Lemma}
\newtheorem{proposition}{Proposition}

\theoremstyle{remark}
\newtheorem{remark}{Remark}

\usepackage{hyperref}
\newcommand{\sn}[1]{S_N\!\left(#1\right)}

\begin{document}
\title{\huge{Joint Chirp Parameter Selection   and Low-Complexity MMSE Receiver Design for AFDM Systems }}
 \author{Ruiyuan Mao, Qu Luo, \textit{Member, IEEE},  Jianguo Li, Fabien Héliot,~\IEEEmembership{Senior Member,~IEEE},   Tianqi Mao, \textit{Member, IEEE}, Pei Xiao,~\IEEEmembership{Senior Member,~IEEE},  Hee Wook Kim,     and Kai Yang, \textit{Member, IEEE}
 \thanks{Ruiyuan Mao and Kai Yang are with the school of information and electronics, Beijing institute of technology, Beijing 100081, China (emails: mry@bit.edu.cn; yangkai@bit.edu.cn). 
 
 Qu Luo, Fabien Héliot and Pei Xiao are with the 5G/6G Innovation Centre, University of Surrey, GU2 7XH Guildford, U.K. (e-mail: \{q.u.luo, f.heliot, p.xiao\}@surrey.ac.uk). 
 
 Jianguo Li is with the School of Cyberspace Science and Technology, Beijing Institute of Technology, Beijing 100081, China (e-mail: jianguoli@bit.edu.cn). 
 
 Tianqi Mao is with State Key Laboratory of Environment
Characteristics and Effects for Near-space, Beijing Institute of Technology,
Beijing 100081, China.(e-mail: maotq@bit.edu.cn).  

Hee Wook Kim is with the Satellite Communication Research Division, Electronics and Telecommunications Research Institute (ETRI), Daejeon 34164, South Korea (email: prince304@etri.re.kr).}
 }
\maketitle

\begin{abstract}
Affine frequency division multiplexing (AFDM) has emerged as a promising waveform against doubly selective channels under high-mobility communication scenarios. 
Optimal chirp parameter ($c_1$) selection and reduced-complexity receiver design in AFDM are essential for achieving satisfactory bit error rate (BER) performance with low computational complexity.  
In this paper, we investigate the joint optimization of chirp-parameter selection and low-complexity minimum mean square error (MMSE)-based receiver design by exploiting the structural characteristics of the AFDM effective channel matrix (ECM). 
First, a simplified BER performance metric is derived by leveraging the diagonal and circulant structure of the discrete affine Fourier transformation (DAFT), based on which a fast circulant-diagonal aggregation (FCDA) algorithm is developed for efficient $c_1$ selection. 
Then, a low-complexity banded MMSE (LC-BMMSE) receiver is developed by constructing a cyclic-banded ECM through path-wise structured sparsification, where banded Cholesky factorization is employed to avoid direct matrix inversion. 
Building upon the proposed BER metric and the LC-BMMSE receiver, a hierarchical-search-based joint chirp parameter and structured sparsification (HS-JCPS) algorithm is further proposed to jointly optimize the chirp parameter and sparsification pattern under a given complexity constraint.
Simulation results demonstrate that the proposed FCDA reduces the search time for the optimal $c_1$ by an order of magnitude compared with using a BER-based criterion.
Moreover, the proposed HS-JCPS algorithm with the LC-BMMSE receiver can identify a near-optimal $c_1$, 
while attaining a superior performance-complexity tradeoff.
\end{abstract}


\section{Introduction}
\label{sec1}
\IEEEPARstart{S}{ixth}-generation wireless systems are expected to support higher data rates and enhanced reliability in emerging high-mobility scenarios, such as vehicle-to-everything (V2X), railway communications, unmanned aerial vehicle (UAV) links, and satellite communications \cite{ref01,Refadd4,ref02,AFDM_JSAC,Refadd3}. In these scenarios, the high relative mobility among transmitters, receivers, and scatterers gives rise to severe Doppler shifts and Doppler spread, while multipath propagation remains significant~\cite{ref03}. As a result, the wireless channel becomes doubly selective, varying over both time and frequency~\cite{ref04}. Such channel dynamics pose fundamental challenges to reliable signal transmission and call for waveform designs with enhanced robustness~\cite{ref02, ref04}.

Orthogonal frequency-division multiplexing (OFDM) has been widely adopted as the dominant waveform in modern wireless systems, primarily due to its ability to divide the available bandwidth into multiple orthogonal narrowband subcarriers, thereby enabling efficient multipath mitigation, flexible resource allocation, and low-complexity frequency-domain equalization~\cite{ref05,Refadd7}. However, these advantages rely critically on the preservation of subcarrier orthogonality, which becomes increasingly vulnerable in doubly selective channels. In such environments, Doppler shifts destroy subcarrier orthogonality and induce severe inter-carrier interference (ICI), resulting in substantial performance degradation and reduced transmission reliability~\cite{ref06,Refadd5}.

\subsection{Related Works}
To address the aforementioned challenges, affine frequency division multiplexing (AFDM) has been recognized as a promising waveform for highly dynamic wireless scenarios~\cite{ref06,ref07,ref08}. 
By modulating symbols on orthogonal chirp subcarriers through the inverse discrete affine Fourier transform (IDAFT), AFDM introduces two tunable chirp parameters, $c_1$ and $c_2$, which provide additional flexibility in waveform design~\cite{ref06}.
When these parameters are properly configured according to the channel delay-Doppler characteristics, AFDM yields a quasi-static and well-structured effective channel matrix (ECM) in the discrete affine Fourier transform (DAFT) domain, thereby enabling full diversity over doubly selective channels~\cite{ref06}. 
In addition, AFDM generalizes both OFDM and orthogonal chirp division multiplexing (OCDM) as special cases, offering a unified waveform framework rather than a completely separate design~\cite{ref02}. 
From an implementation perspective, its transceiver can be efficiently realized by incorporating low-complexity diagonal pre- and post-chirp operations into conventional fast Fourier transform (FFT)-based architectures. Benefiting from these attractive properties, AFDM has drawn increasing attention from both academia and industry~\cite{xu2025afdm,Refadd7,ref09,ref10,ref11,ref12,ref13,ref14,ref15,IM_AFDM_TWC,Refadd8,cao2025agile}.

To date, AFDM has been extensively investigated in a wide range of areas, including integrated sensing and communication (ISAC)~\cite{ref09,ref10,yin2025ofdm}, index modulation (IM)~\cite{IM_AFDM_TWC,sui2025generalized}, peak-to-average power ratio (PAPR) reduction~\cite{ref11,ref12}, parameter optimization~\cite{ref13}, and low-complexity receiver design~\cite{ref14,ref15}. 
Among these directions, parameter optimization and low-complexity receiver design are particularly important, since they directly affect the bit error rate (BER) performance and computational complexity of AFDM systems. 
For parameter optimization, an adaptive Agile-AFDM scheme was proposed in~\cite{cao2025agile}, where the chirp parameters are dynamically selected for each transmission block according to real-time channel conditions, thereby improving the overall system performance. 
For low-complexity receiver design, several representative approaches have been reported. In~\cite{ref14}, low-complexity equalizers were developed for AFDM over doubly dispersive channels by exploiting the structured ECM. Message-passing-based AFDM detection was investigated in~\cite{ref15} to alleviate the detection burden through iterative probabilistic inference. In~\cite{R11}, a time-domain zero-padding AFDM scheme with two-stage iterative minimum mean square error (MMSE) detection was proposed, while~\cite{R13} developed a joint sparse-graph framework for enhanced MIMO-AFDM receiver design. 
For IM-AFDM, a double-layer message-passing detector was proposed in~\cite{IM_AFDM_add1} to jointly infer the modulation symbols and index-activation patterns, thereby avoiding exhaustive maximum likelihood (ML) detection and achieving a favorable performance--complexity tradeoff. 
More recently, a deep-learning-aided expectation propagation (EP) turbo detector was developed in~\cite{DL_AFDM_add2}, where the sparse and quasi-banded DAFT-domain channel structure was exploited through block decomposition and structured matrix factorization to reduce the EP detection complexity from cubic to linear order while maintaining near-ML performance. 
More closely related to this paper, MMSE-based  receiver design has been extensively studied by exploiting the path separability and structural sparsity of the ECM~\cite{ref14,ref16,R11,di2025parameter}. 

Existing studies on MMSE-based receivers mainly focus on two aspects. 
On the one hand, the ECM structure can be tailored through chirp-parameter design to improve MMSE receiver performance \cite{di2025parameter,ref16}. For example, the authors in \cite{ref16} investigated the performance limits of MMSE detection by deriving a theoretical lower bound on the BER, and further proposed an effective parameter-selection strategy by assuming a low-Doppler regime.
On the other hand, several works attempted to sparsify the ECM and develop low-complexity MMSE receivers \cite{ref14,R11}. 
Specifically, the authors in~\cite{ref14} introduced nulling subcarriers in the DAFT domain to obtain a banded ECM at the cost of spectral efficiency, based on which a low-complexity detection was developed using LDL factorization.
In contrast to~\cite{ref14}, the authors in~\cite{R11} replaced the chirp-periodic prefix (CPP) with a zero sequence to construct a banded ECM, and developed a two-stage iterative MMSE detection algorithm.

\subsection{Motivation and Contributions}
Although existing AFDM studies have made important progress in performance enhancement \cite{di2025parameter,ref16} and complexity reduction \cite{ref06,R11,ref14}, several open issues remain when MMSE detection is considered.  
Firstly, most parameter-selection studies aimed at improving performance are based on the full ECM, while the practical complexity burden is not explicitly taken into account \cite{cao2025agile,di2025parameter,ref17}.
Although preliminary work on AFDM parameter selection for MMSE-based detection has been reported in \cite{ref16}, it implicitly assumes a low-mobility scenario.   
Moreover, since AFDM performance is jointly determined by both chirp parameter selection and receiver design, a joint optimization of these two aspects is essential for achieving an overall performance improvement.
However, existing studies typically address these two aspects separately, and their joint optimization for MMSE-based AFDM detection remains largely unexplored.

Motivated by the above observations, we investigate the joint optimization of chirp-parameter selection and low-complexity MMSE receiver design by fully exploiting the structural characteristics of the AFDM  ECM.  

The main  contributions of this paper are summarized as follows:
\begin{enumerate}
\item
We propose a fast circulant-diagonal aggregation (FCDA) algorithm for near-optimal chirp-parameter selection under a general MMSE detector. Specifically, we derive a simplified BER performance metric by leveraging the diagonal and circulant structural characteristics induced by the AFDM chirp transformation, which substantially lowers the computational complexity while still maintaining effectiveness in high-Doppler scenarios. Subsequently, based on the proposed BER metric, the FCDA algorithm is developed for efficient $c_1$ selection.

\item A sparsity-exploiting low-complexity banded MMSE (LC-BMMSE) receiver is proposed by exploiting the sparsity of the AFDM ECM. Specifically, we first construct a cyclic-banded ECM structure and then perform path-wise structured sparsification over the multipath components to explicitly control the receiver complexity. This banded ECM structure facilitates the design of the LC-BMMSE receiver. In particular, the proposed LC-BMMSE receiver leverages the cyclic band property of the sparsified ECM to compress the corresponding autocorrelation matrix into a tighter banded form. Based on this structure, banded Cholesky factorization is employed to avoid the high computational cost of direct matrix inversion.

\item Building upon the proposed FCDA and LC-BMMSE receiver, we further develop a hierarchical-search-based joint chirp parameter and structured sparsification (HS-JCPS) algorithm for joint chirp-parameter selection and structured-sparsification optimization to identify the optimal sparsification strategy under a given complexity constraint.  Specifically, a lower bound on the MMSE performance for the sparsified ECM is derived which serves as a computationally efficient BER metric, followed by the sparsification pattern design. Furthermore, to reduce the search space of both the $c_1$ values and the sparsification patterns, the proposed HS-JCPS algorithm recursively determines the chirp parameter together with the corresponding sparsification pattern.

 \item We conduct extensive simulations to validate the proposed FCDA algorithm, LC-BMMSE receiver, and HS-JCPS algorithm. The results show that the proposed FCDA can reduce the $c_1$ search time by an order of magnitude. In addition, the proposed HS-JCPS with LC-BMMSE receiver is able to identify a near-optimal $c_1$ and pruning width $\mathbf{w}$, achieving an approximately $1$ dB performance gain at a BER of $10^{-3}$.
Overall, the joint optimization of $c_1$ selection and low-complexity receiver design yields a superior performance-complexity tradeoff.
\end{enumerate}

\subsection{Organization and Notations}
The rest of this paper is organized as follows. 
Section \ref{sec2} describes the preliminary concepts of doubly selective channels and AFDM.
In Section \ref{sec3}, a fast optimal AFDM parameter selection for full ECM is presented. 
  Section IV presents the proposed LC-BMMSE
receiver along with the HS-JCPS algorithm. 
Simulation results are presented  in Section \ref{sec6}. 
Section \ref{sec7} concludes the paper.

\textit{Notations}:
In this paper, vectors and matrices are denoted by lowercase and uppercase boldface letters, respectively.
Subscripts in upright text indicate names, while those in italics indicate variables.
$(\cdot)^\mathcal{T}$, $(\cdot)^\mathcal{H}$, $(\cdot)_N$, $\lfloor\cdot\rfloor$ and $\odot$ denote the transpose, the Hermitian transpose, the modulo-$N$, the floor operation and the Hadamard product operation, respectively.
$\mathbf{I}_N$ denotes an $N\times N$-dimensional identity matrix.
$\mathbb{C}^{N\times 1}$ and $\mathbb{Z}^P_+$ denote the $N$-dimensional complex vector space and the set of positive integer vectors of dimension $P$, respectively.

\begin{figure*}[t]
	\centering
	\includegraphics[width = 2\columnwidth]{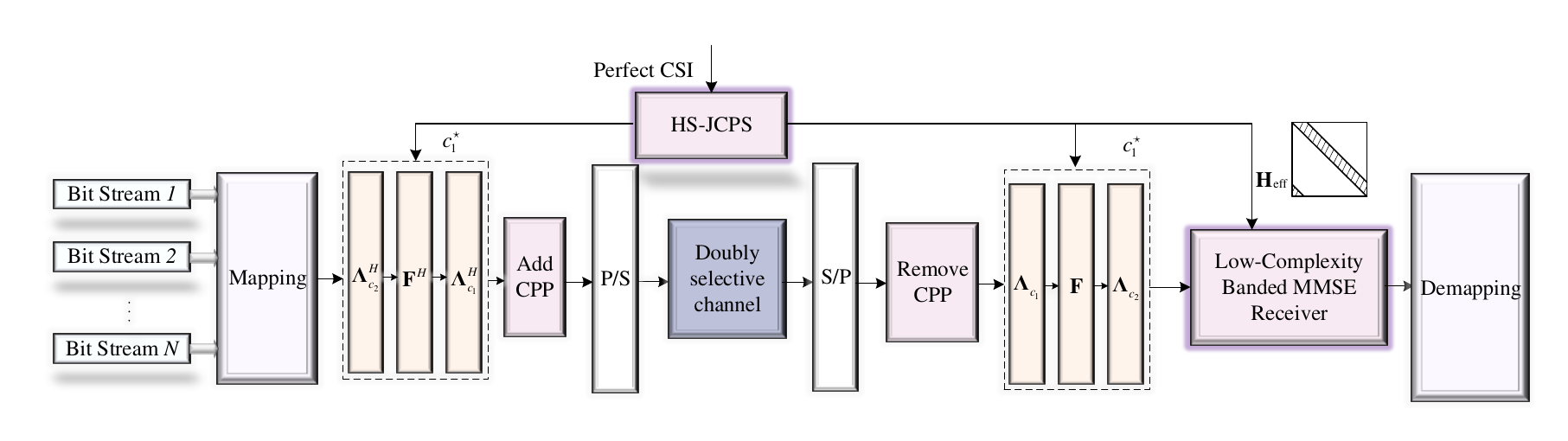}
	\caption{Block diagram of the AFDM system with the proposed HS-JCPS and low-complexity banded MMSE receiver.}
	\label{fig:systemModel}
\end{figure*}

\section{Preliminaries}
\label{sec2}
In this section, we introduce the digital implementation of the AFDM system as shown in Fig.~\ref{fig:systemModel} and establish the input-output relationship in the affine domain.
Then, we formulate the optimization problem that maximizes the MMSE performance subject to AFDM parameters.

\subsection{AFDM Modulation}
Let $\mathbf{s}=[s_1,s_2,\cdots,s_N]^\mathcal{T}\in\mathbb{C}^{N\times 1}$ denote the transmitted 
vector in the DAFT domain, where $\{s_n\}_{n=1}^N$ are quadrature amplitude modulation 
(QAM) symbols satisfying $\mathbb{E}\{\lvert s_n\rvert^2\}=\mathrm{E}_\mathrm{s}$.
The modulated time-domain symbol vector is obtained by applying the IDAFT to $\mathbf{s}$:
\begin{align}
	\begin{aligned}
		\mathbf{x}=\mathbf{A}^\mathcal{H}\mathbf{s},
		\label{eq:Mod_AFDM}
	\end{aligned}
\end{align}
where $\mathbf{A} = \boldsymbol{\Lambda}_{c_2}\mathbf{F}\boldsymbol{\Lambda}_{c_1}$ is the DAFT matrix, $c_1  $ and $c_2  $ are the AFDM parameters,
$\boldsymbol{\Lambda}_{c}=\operatorname{diag}\left\{e^{-j2\pi c n^2},n=0,1,\cdots,N-1\right\}$, and $\mathbf{F}$ is the DFT matrix with entries $\frac{1}{\sqrt{N}}e^{-j\frac{2\pi}{N}mn}$.

In AFDM, a prefix is introduced to combat inter-symbol interference caused by multipath and renders the channel matrix circulant.
Unlike OFDM, which relies on a CP, AFDM exploits the chirp periodicity of the DAFT and appends a CPP. 
With this construction, the linear convolution with the multipath channel can be equivalently represented as a circulant convolution within the AFDM symbols.

\subsection{Channel}
Consider a doubly selective channel with $P$ paths.
The time-domain channel impulse response at time $n$ and delay $l$ is given by 
\begin{align}
	\begin{aligned}
		h(n,l) = \sum_{i=1}^P h_i e^{-j\frac{2\pi}{N} \nu_i n} \delta(l-l_i),
	\end{aligned}
\end{align}
where $\delta(\cdot)$ is the Dirac delta function. 
$h_i$, $l_i$ and $\nu_i$ are the complex gain, the integer delay and digital Doppler shift of the $i$th path, respectively.
The digital Doppler shift is obtained by normalizing the Doppler frequency $f_i$ with respect to the subcarrier spacing $\Delta f$ and can be expressed as $\nu_i =\frac{f_i}{\Delta f} = \alpha_i + \beta_i$, 
where $\alpha_i \in[-\alpha_\mathrm{max},\alpha_\mathrm{max}]$ is the integer part of $\nu_i$, $\beta_i \in(-\frac{1}{2},\frac{1}{2}]$ is its fractional part , and $\alpha_\mathrm{max}$ is the maximum integer Doppler.
The time-domain channel impulse response can be represented in the matrix form as
\begin{align}
	\begin{aligned}
		\mathbf{H} = \sum_{i=1}^P h_i \boldsymbol{\Gamma}_{\mathrm{CPP}_i}\boldsymbol{\Delta}_{\nu_i} \boldsymbol{\Pi}^{l_i},
		\label{eq:DDC_AFDM}
	\end{aligned}
\end{align}
where $\boldsymbol{\Delta}_{\nu_i}=\operatorname{diag}\left\{e^{-j\frac{2\pi}{N} \nu_i n},n=0,1,\cdots,N-1\right\}$ models the Doppler of the $i$th path, $\boldsymbol{\Pi}^{l_i}$ denotes the forward cyclic-shift matrix \cite{ref06}, and $\boldsymbol{\Gamma}_{\mathrm{CPP}_i}$ is the effective CPP matrix with its expression given by 
\begin{align}\label{eq:CPP}
	\mathsmaller{\boldsymbol{\Gamma}_{\mathrm{CPP}_i}=
	\operatorname{diag}\left(\begin{cases}
		e^{\mathsmaller{-j 2 \pi  c_1\left(N^2-2N\left(l_i-n\right)\right)}}, & n<l_i \\
		1, & n \geq l_i
	\end{cases} \right).}
\end{align}
From \eqref{eq:CPP}, it follows that when $2Nc_1$ is an integer and $N$ is even, $\boldsymbol{\Gamma}_{\mathrm{CPP}_i}=\mathbf{I}_N$.

\subsection{AFDM Demodulation}
\label{secdemo}
At the receiver end, after removing the CPP, the received signal in the DAFT domain is obtained by further applying the DAFT, i.e.,  
\begin{equation}\label{eq:Demod_AFDM}
		\mathbf{y} = \mathbf{A}\mathbf{H}\mathbf{x}+\mathbf{n},
\end{equation}
where $\mathbf{n}\sim\mathcal{CN}(\bm{0}_N,\sigma_\mathrm{n}^2\mathbf{I}_N)$ denotes the complex additive white Gaussian noise vector.

Substituting \eqref{eq:Mod_AFDM} and \eqref{eq:DDC_AFDM} into \eqref{eq:Demod_AFDM}, 
the input-output relationship in the DAFT domain in the matrix form is given by
\begin{align}\label{eq:IOrelation}
	\begin{aligned}
		\mathbf{y}= & \sum_{i=1}^Ph_i\underbrace{\boldsymbol{\Lambda}_{c_2}\mathbf{F}\boldsymbol{\Lambda}_{c_1}\boldsymbol{\Delta}_{\nu_i} \boldsymbol{\Pi}^{l_i}\boldsymbol{\Lambda}_{c_1}^\mathcal{H}\mathbf{F}^\mathcal{H}\boldsymbol{\Lambda}_{c_2}^\mathcal{H}}_{\mathbf{H}_{\mathrm{eff},i}}\mathbf{s}+\mathbf{n}\\
		=&\ \mathbf{H}_{\mathrm{eff}}\ \mathbf{s}+\mathbf{n},
	\end{aligned}
\end{align}
where $\mathbf{H}_{\mathrm{eff}} \triangleq \sum_{i=1}^P h_i \mathbf{H}_{\mathrm{eff},i}$ 
is the full ECM, i.e., the matrix that retains all channel-induced entries without sparsification. 
The $(p,q)$-th entry of 
$\mathbf{H}_{\mathrm{eff},i}$ is given by \cite{ref06}
\begin{align}
	\begin{aligned}
		[\mathbf{H}_{\mathrm{eff},i}]_{p,q}=&\frac{1}{N}
		e^{j\frac{2\pi}{N}\phi_i(p,q)}\sn{p-q+\xi_i+\beta_i},\\
		& p,q\in\{0,\dots,N-1\},
		\label{eq:Hdef}
	\end{aligned}
\end{align}
where $\phi_i(p,q)=N c_1 l_i^2- q l_i + N c_2 (q^2-p^2)$, and
\begin{align}
	\begin{aligned}
		\sn{p-q+\xi_i+\beta_i}\triangleq  \frac{e^{-j2\pi\left(p-q+\xi_i+\beta_i\right)}-1}{e^{-j\frac{2\pi}{N}\left(p-q+\xi_i+\beta_i\right)}-1},
		\label{Eq:DFTkernel}
	\end{aligned}
\end{align}
where $	\xi_i\triangleq (\alpha_i+2Nc_1 l_i)_N$.
For fractional Doppler case, i.e., $\beta_i\neq0$, $\left\lvert [\mathbf{H}_{\mathrm{eff},i}]_{p,q} \right\rvert$ attains its peak at $q=(p+\xi_i)_N$ and decays as $q$  moves away from $\xi_i$, as shown in Fig.~\ref{fig:SparseECM}.
To facilitate low-complexity receiver implementation, we retain only the entries within a local neighborhood of width $2w_i+1$ centered at $(p+\xi_i)_N$, where $w_i$ denotes the pruning width of the $i$th path. 
Namely,
\begin{align}\label{eq:purning}
    \begin{aligned}
        [\tilde{\mathbf{H}}_{\mathrm{eff},i}]_{p,q}=\begin{cases}
            [\mathbf{H}_{\mathrm{eff},i}]_{p,q} \quad  & \left(p+\xi_i-w_i\right)_N \leq q  \\ & \leq \left(p+\xi_i+w_i\right)_N\\
            0\quad &\mathrm{otherwise}
        \end{cases}.
    \end{aligned}
\end{align}
By summing the pruned matrices of all paths, the sparsified ECM is obtained as
\begin{equation}
	\tilde{\mathbf{H}}_\mathrm{eff} = \sum_{i=1}^{P}h_i\tilde{\mathbf{H}}_{\mathrm{eff},i}.
	\label{eq:sparECM}
\end{equation}
The corresponding sparsification perturbation matrix (SPM) is given by $\Delta \mathbf{H}_\mathrm{eff}= \mathbf{H}_\mathrm{eff}-\tilde{\mathbf{H}}_\mathrm{eff}$.
Its average power is quantified by $\sigma_\Delta^2 = \frac{\operatorname{Tr}\!\left(\Delta\mathbf{H}_\mathrm{eff}\Delta\mathbf{H}_\mathrm{eff}^\mathcal{H}\right)}{N}
$.

\begin{figure}
	\centering
	\includegraphics[width = 0.9\columnwidth]{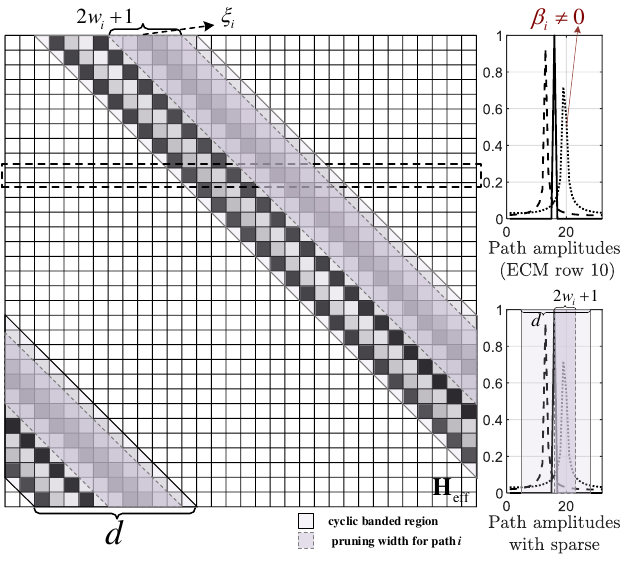}
	\caption{Structure of ECM and sparsification operation with normalized time delays $[1,2,3]$, $c_1=3/2N$ and  $N=32$.}
	\label{fig:SparseECM}
\end{figure}

\subsection{Problem Statement} 
Due to the truncation operation, the sparsity pattern of $\tilde{\mathbf{H}}_\mathrm{eff}$ becomes circularly banded.
For the $i$th path, $w_i$ controls the local retained neighborhood around the dominant coupling center $q=(p+\xi_i)_N$. 
After the pruned components of all paths are superimposed, the sparsity pattern of $\tilde{\mathbf H}_{\rm eff}$ is determined by the union of the retained supports of all paths. 
Accordingly, $d$ denotes the overall DAFT-domain truncation width of the sparsified ECM, as illustrated in Fig.~\ref{fig:SparseECM}. 
When the path centers are ordered within the considered cyclic interval as $\xi_1\le\xi_2\le\cdots\le\xi_P$, the aggregate retained region extends from $\xi_1-w_1$ to $\xi_P+w_P$, yielding $d=w_1-\xi_1+w_P+\xi_P$.
Both $d$ and $c_1$ jointly determine the sparsity of the sparsified ECM.  
In general, a smaller $d$ yields a sparser ECM, which can be further exploited to reduce detection complexity. 
However, overly aggressive truncation risks discarding significant channel energy, incurring a non-negligible performance penalty.
This fundamental complexity-performance trade-off motivates the joint design of the AFDM parameters and a low-complexity receiver. 
Specifically, we consider the joint optimization of the AFDM parameter pair 
$(c_1, c_2)$ and the design of a sparsity-aware MMSE receiver, with the aim of minimizing the BER performance.

\section{Fast Optimal AFDM Parameter Selection for Full ECM}
\label{sec3}
This section presents a fast optimal AFDM parameter selection scheme 
based on an MMSE receiver under the full ECM assumption, which serves as 
the foundation for the subsequent sparsified MMSE receiver design. 
Specifically, we first analyze the impact of the AFDM parameter $c_1$ on 
BER performance under the full ECM, i.e., $d = N$. Building on this 
analysis, we then propose the FCDA algorithm to efficiently identify the optimal 
$c_1$ while substantially reducing the search complexity.

\subsection{BER Analysis under MMSE Detection}
Assuming that an MMSE receiver is employed, the estimated symbol 
vector is expressed as
\begin{equation}
        \hat{\mathbf{s}}=(\underbrace{\mathbf{G}+\frac{\sigma_\mathrm{n}^2}{\mathrm{E}_\mathrm{s}}\mathbf{I}_N}_{\boldsymbol{\Psi}})^{-1}\mathbf{H}_{\mathrm{eff}}^\mathcal{H}\left(\mathbf{H}_{\mathrm{eff}}\mathbf{s}+\mathbf{n}\right),
		\label{eq:est_signal}
\end{equation}
where $\mathbf{G}=\mathbf{H}_{\mathrm{eff}}^\mathcal{H}\mathbf{H}_{\mathrm{eff}}$ refers to the Gram matrix of the full ECM, and $\boldsymbol{\Psi}$ represents the regularized matrix of $\mathbf{G}$. 
The signal-to-interference-plus-noise ratio (SINR) of the $k$th stream 
of $\hat{\mathbf{s}}$ is given by \cite{P1}
\begin{equation}
	\operatorname{SINR}_k = \frac{\mathrm{E}_\mathrm{s}\lvert[\boldsymbol{\Psi}^{-1}\mathbf{G}]_{k,k}\rvert^2}{\mathrm{E}_\mathrm{s}\sum_{j\neq k}\lvert[\boldsymbol{\Psi}^{-1}\mathbf{G}]_{k,j}\rvert^2+\sigma_\mathrm{n}^2\boldsymbol{\Psi}^{-1}\mathbf{G}\boldsymbol{\Psi}^{-1}}. 
	\label{eq:SINRk}
\end{equation}
With perfect channel state information (CSI) available at the receiver, 
\eqref{eq:SINRk} simplifies to \cite{P3}
	\begin{align}
		\begin{aligned}
			\operatorname{SINR}_k = \frac{[\mathbf{G}_0]_{k,k}}{1-[\mathbf{G}_0]_{k,k}},
			\label{eq:SINR_MMSE}
		\end{aligned}
	\end{align}
where $[\mathbf{G}_0]_{k,k}=[(\mathbf{I}_N+\frac{\sigma_\mathrm{n}^2}{\mathrm{E}_\mathrm{s}}\mathbf{G}^{-1})^{-1}]_{k,k}$.
Then, under $M$-QAM mapping and Gray code mapping, the average BER can be calculated as \cite{P2} 
\begin{equation}
	\begin{aligned}
		P_\mathrm{A} = \frac{1}{N}\sum_{k=1}^{N}\underbrace{\mathsmaller{a_M\operatorname{Q}\left(\sqrt{\frac{3\operatorname{SINR}_k}{M-1}}\right)-a_M^2\operatorname{Q}^2\left(\sqrt{\frac{3\operatorname{SINR}_k}{M-1}}\right)}}_{\operatorname{BER}_k},
	\end{aligned}
	\label{eq:Calber}
\end{equation}
where $a_M=4\frac{\sqrt{M}-1}{\sqrt{M}}$, $\operatorname{BER}_k$ refers to the BER of the $k$th stream of $\hat{\mathbf{s}}$, and $\operatorname{Q}(x)$ denotes the $\operatorname{Q}$ function with its expression given by
$\operatorname{Q}(x)=\frac{1}{\sqrt{2\pi}}\int_x^\infty e^{\frac{-t^2}{2}}dt$.

\subsection{FCDA Algorithm for Optimal \texorpdfstring{$c_1$}{c1} Selection}
We start by looking into  how  AFDM parameters   $\left(c_1, c_2\right)$ affect  the BER performance, i.e., $\operatorname{BER}_k$. 
Unlike existing counterparts that focus on the small-Doppler regime \cite{ref16}, we consider the more general large fractional-Doppler case, where the coupling structure of ECM becomes more complex.
Specifically, $\mathbf{G}$ can be equivalently decomposed into $\mathbf{G}=\boldsymbol{\Lambda}_{c_2}\mathbf{M}\boldsymbol{\Lambda}_{c_2}^\mathcal{H}$, where $\mathbf{M}=\mathbf{F}\boldsymbol{\Lambda}_{c_1}\mathbf{H}^\mathcal{H}\mathbf{H}\boldsymbol{\Lambda}_{c_1}^\mathcal{H}\mathbf{F}^\mathcal{H}$.
This reformulation makes the roles of $c_1$ and $c_2$ in determining the BER more explicit, which leads to Proposition~\ref{pro1}.

\begin{proposition}\label{pro1}
	From (15), the impact of the AFDM parameter $c_1$ on $\operatorname{BER}_k$ is characterized by the $k$th diagonal entry of $\left(\mathbf{M}+\frac{\sigma_\mathrm{n}^2}{\mathrm{E}_\mathrm{s}}\mathbf{I}_N \right)^{-1}$.
Note that, for ideal transmission, such as without nonlinear PAPR distortion, $c_2$ generally does not affect the $\operatorname{BER}_k$ when both the transmitter and the receiver employ the same $c_2$.
\end{proposition}
\begin{proof}
	Since $\boldsymbol{\Lambda}_{c_2}$ is a unitary diagonal matrix, $\left[\mathbf{G}_0\right]_{k,k}$ can be derived as
	\begin{align}
		\begin{aligned}
			\left[\mathbf{G}_0\right]_{k,k} = &\left[\left(\mathbf{I}_N+\frac{\sigma_\mathrm{n}^2}{\mathrm{E}_\mathrm{s}}\boldsymbol{\Lambda}_{c_2}\mathbf{M}^{-1}\boldsymbol{\Lambda}_{c_2}^\mathcal{H}\right)^{-1}\right]_{k,k}\\
			=&\boldsymbol{\Lambda}_{c_2}\left[\left(\mathbf{I}_N+\frac{\sigma_\mathrm{n}^2}{\mathrm{E}_\mathrm{s}}\mathbf{M}^{-1}\right)^{-1}\right]_{k,k}\boldsymbol{\Lambda}_{c_2}^\mathcal{H} \\
			= &1-\frac{\sigma_\mathrm{n}^2}{\mathrm{E}_\mathrm{s}}\left[\left(\mathbf{M}+\frac{\sigma_\mathrm{n}^2}{\mathrm{E}_\mathrm{s}}\mathbf{I}_N\right)^{-1}\right]_{k,k}.
			\label{eq:ind_c2}
		\end{aligned}
	\end{align}
Substituting \eqref{eq:ind_c2} into \eqref{eq:SINR_MMSE} yields
	\begin{align}
		\begin{aligned}
			\operatorname{SINR}_k = \frac{\mathrm{E}_\mathrm{s}}{\sigma_\mathrm{n}^2\left[\left(\mathbf{M}+\frac{\sigma_\mathrm{n}^2}{\mathrm{E}_\mathrm{s}}\mathbf{I}_N\right)^{-1}\right]_{k,k}}-1.
			\label{eq:SINR_kforM}
		\end{aligned}
	\end{align}
It follows from \eqref{eq:SINR_kforM} that $c_1$ affects $\operatorname{SINR}_k$ through the $k$th diagonal entry of $\left(\mathbf{M}+\frac{\sigma_\mathrm{n}^2}{\mathrm{E}_\mathrm{s}}\mathbf{I}_N \right)^{-1}$, while $\operatorname{BER}_k$ is uniquely determined by $\operatorname{SINR}_k$.
Moreover, since $\mathbf{M}$ is independent of $c_2$, $\operatorname{BER}_k$ is also independent of $c_2$, which completes the proof.
\end{proof}

Based on Proposition~\ref{pro1}, we propose the FCDA algorithm for efficiently selecting the optimal $c_1$.
Specifically, we define $\mathbf{H}_\mathrm{t}$ as the regularized inverse channel Gram matrix with respect to $\mathbf{H}$, i.e., $\mathbf{H}_\mathrm{t} = \left(\mathbf{H}^\mathcal{H}\mathbf{H}+\frac{\sigma_\mathrm{n}^2}{\mathrm{E}_\mathrm{s}}\mathbf{I}_N\right)^{-1}$, then $\operatorname{SINR}_k$ can be rewritten as
\begin{align}
	\begin{aligned}
		\operatorname{SINR}_k = \frac{\mathrm{E}_\mathrm{s}}{\sigma_\mathrm{n}^2\left[\mathbf{F}\boldsymbol{\Lambda}_{c_1}\mathbf{H}_\mathrm{t}\boldsymbol{\Lambda}_{c_1}^\mathcal{H}\mathbf{F}^\mathcal{H}\right]_{k,k}}-1.
		\label{eq:SINR_krew}
	\end{aligned}
\end{align}
To further isolate the effect of $c_1$ on $\operatorname{SINR}_k$, 
$\mathbf{H}_\mathrm{t}$ is decomposed into a diagonal part and an 
off-diagonal part as
\begin{equation}
	\mathbf{H}_\mathrm{t} = \underbrace{\operatorname{diag}\left(\mathbf{H}_\mathrm{t}\right)}_{\mathrm{diagonal \ part}} + \underbrace{\sum_{u=1}^{N-1} \operatorname{diag}\left(\mathbf{h}_u\right)\boldsymbol{\Pi}^u}_{\mathrm{off-diagonal \ part}},
	\label{eq:Ht}
\end{equation}
where $\mathbf{h}_u$ denotes the $u$th lower circulant-diagonal vector of $\mathbf{H}_{\mathrm t}$, i.e., $\operatorname{diag}\!\left(\mathbf{H}_{\mathrm t}\left(\boldsymbol{\Pi}^u\right)^\mathcal{H}\right)$.
Consequently, the $k$th diagonal element of $\mathbf{F}\boldsymbol{\Lambda}_{c_1}\mathbf{H}_\mathrm{t}\boldsymbol{\Lambda}_{c_1}^\mathcal{H}\mathbf{F}^\mathcal{H}$ can be expressed as
\begin{align}
	\begin{aligned}
	   \left[\mathbf{F}\boldsymbol{\Lambda}_{c_1}\mathbf{H}_\mathrm{t}\boldsymbol{\Lambda}_{c_1}^\mathcal{H}\mathbf{F}^\mathcal{H}\right]_{k,k}=\left[\mathbf{H}_\mathrm{t}\right]_{k,k}+\left[\mathbf{L}\right]_{k,k},
	\end{aligned}
\end{align}
where $\mathbf{L}$ refers to the off-diagonal part of $\mathbf{H}_\mathrm{t}$ after the similarity transformation, which is given by
\begin{align}
	\begin{aligned}
		\mathbf{L} =& \sum_{u=1}^{N-1}\underbrace{\mathbf{F}\boldsymbol{\Lambda}_{c_1}\left(\operatorname{diag}\left(\mathbf{h}_u\right)\boldsymbol{\Pi}^u\right)\boldsymbol{\Lambda}_{c_1}^\mathcal{H}\mathbf{F}^\mathcal{H}}_{\mathbf{L}_{u}}.
		\label{eq:Ld}
	\end{aligned}
\end{align}
It is noteworthy that the diagonal part of $\mathbf{H}_\mathrm{t}$ remains invariant under the unitary similarity transformation induced by $\mathbf{F}\boldsymbol{\Lambda}_{c_1}$, whereas the off-diagonal circulant-diagonal components are mixed by the transformation and generate additional contributions to the main diagonal. 
Therefore, the $k$th diagonal entry after the transformation consists of the original diagonal term $\left[\mathbf{H}_\mathrm{t}\right]_{k,k}$ and the contribution $\left[\mathbf{L}\right]_{k,k}$ induced by the transformed off-diagonal part.
From \eqref{eq:SINR_krew}, the variation of $\operatorname{SINR}_k$ with respect to $c_1$ is determined by $[\mathbf{L}]_{k,k}$, which represents the contribution of the transformed off-diagonal coupling to the main diagonal.
To capture the average BER $P_\mathrm{A}\left(c_1\right)$ over all streams, we aggregate this diagonal contribution by defining the total diagonal energy of $\mathbf{L}$ as  
\begin{equation}
	P_\mathrm{e}(c_1) = \sum_{k=1}^N \left\lvert \left[\mathbf{L}\right]_{k,k} \right\rvert^2.
	\label{eq:Pe}
\end{equation}

\begin{remark}
\textit{
The metric $P_\mathrm{e}(c_1)$ is effective for comparing candidate $c_1$ values under identical channel realization, noise power, modulation, and MMSE receiver settings. 
Specifically, the off-diagonal circulant-diagonal components of $\mathbf H_t$ describe the interference coupling among different DAFT-domain symbols. The chirp parameter $c_1$ changes the phase alignment of these coupling components, and the transformed term $[\mathbf L]_{k,k}$ represents the effective coupling contribution entering the $k$th stream's post-MMSE SINR. Therefore, $P_\mathrm{e}(c_1)$ measures the aggregate energy of the $c_1$-dependent DAFT-domain interference coupling over all streams. A smaller $P_\mathrm{e}(c_1)$ indicates weaker effective coupling, which generally corresponds to higher post-MMSE SINR and lower BER.}
\end{remark}

In what follows, we decompose $\left[\mathbf{L}\right]_{k,k}$ into a computationally tractable form, so that $P_\mathrm{e}(c_1)$ can be evaluated more efficiently.
By exploiting the circulant-diagonal structure of $\operatorname{diag}\left(\mathbf{h}_u\right)\boldsymbol{\Pi}^u$, the similarity transformation $\boldsymbol{\Lambda}_{c_1}(\cdot)\boldsymbol{\Lambda}_{c_1}^\mathcal{H}$ can be greatly simplified.
Specifically, the original matrix multiplication reduces to an element-wise modulation of the circulant-diagonal vector $\mathbf{h}_u$, i.e.,
\begin{align}
\boldsymbol{\Lambda}_{c_1}\left(\operatorname{diag}\left(\mathbf{h}_u\right)\boldsymbol{\Pi}^u\right)\boldsymbol{\Lambda}_{c_1}^\mathcal{H}= \operatorname{diag}\!\left(\mathbf{h}_u \odot \mathbf{q}_u\right)\boldsymbol{\Pi}^u,
\label{eq:Lambda_Hmask_Lambda}
\end{align}
where $\mathbf{q}_u$ denotes the element-wise phase modulation vector induced by the similarity transformation on the $u$th circulant-diagonal component, and its entries are given by
\begin{equation} 
\left[\mathbf{q}_u\right]_\mathsmaller{n}
= e^\mathsmaller{j2\pi c_1\left(\left(\left(n-u\right)_N\right)^2-n^2\right)}.
\label{eq:qk_def}
\end{equation}
\eqref{eq:Lambda_Hmask_Lambda} indicates that the chirp matrix $\boldsymbol{\Lambda}_{c_1}$ affects the $u$th circulant diagonal only through a per-entry phase weighting, i.e., the original circulant-diagonal vector $\mathbf{h}_u$ is element-wise modulated by $\mathbf{q}_u$.
Substituting \eqref{eq:Lambda_Hmask_Lambda} into $\mathbf{L}_{u}$, we obtain
\begin{equation}
	\mathbf{L}_{u} = \mathbf{F}\operatorname{diag}\!\left(\mathbf{h}_u \odot \mathbf{q}_u\right)\boldsymbol{\Pi}^u\mathbf{F}^\mathcal{H}.
\end{equation}
Since the surrogate metric $P_\mathrm{e}(c_1)$ in \eqref{eq:Pe} depends on the diagonal elements of $\mathbf{L}$, only the $k$th row and $k$th column need to be calculated, i.e.,
\begin{align}
	\begin{aligned}
		\left[\mathbf{L}_{u}\right]_{k,k}
		&= \mathbf{f}_k  \operatorname{diag}(\mathbf{h}_u\odot\mathbf{q}_u)\boldsymbol{\Pi}^u\mathbf{f}_k^\mathcal{H} \\
		&= \sum_{n=0}^{N-1} \left[\mathbf{h}_u\odot\mathbf{q}_u\right]_n\, 
		\left[\mathbf{f}_k\right]_{n}\left[\mathbf{f}_k^\mathcal{H}\right]_{(n-u)_N} \\
		&= \frac{1}{N}\sum_{n=0}^{N-1}\left[\mathbf{h}_u\odot\mathbf{q}_u\right]_n\,e^{j\frac{2\pi}{N}(k-1)u},
		\label{eq:diag_term_fft}
	\end{aligned}
\end{align}
where $\mathbf{f}_k$ denotes the $k$th row of $\mathbf{F}$.
Substituting \eqref{eq:diag_term_fft} into \eqref{eq:Ld} and then into \eqref{eq:Pe}, $P_\mathrm{e}(c_1)$ can be expressed as
\begin{equation}
	P_\mathrm{e}(c_1)
	=\sum_{k=1}^{N}\left|
	\frac{1}{N}\sum_{u=1}^{N-1}\sum_{n=0}^{N-1}\left[\mathbf{h}_u\odot\mathbf{q}_u\right]_n e^{j\frac{2\pi}{N}(k-1)u}
	\right|^2.
	\label{eq:Pe_DFT_form}
\end{equation}
\begin{remark}\label{Remark1}
\textit{
Let $x_u \triangleq \sum_{n=0}^{N-1}\!\left[\mathbf{h}_u\odot\mathbf{q}_u\right]_n$. 
The term inside the absolute value in \eqref{eq:Pe_DFT_form} is the $N$-point inverse DFT of the sequence $\{x_u\}_{u=0}^{N-1}$, where the direct current (DC) component $\frac{x_0}{N}$ is excluded. 
Even if $x_0\neq 0$, its time-domain contribution is a constant offset, whereas the non-DC inverse DFT component has zero mean because $\sum_{k=1}^{N} e^{j\frac{2\pi}{N}(k-1)u}=0$ for all $u$.
Hence, the constant offset is orthogonal to the non-DC part and does not affect its energy. 
Therefore, by Parseval's theorem \cite{B1},  one has
\begin{equation}\label{eq:Pasv}
\sum_{k=1}^{N}\left|
\frac{1}{N}\sum_{u=1}^{N-1} x_u e^{j\frac{2\pi}{N}(k-1)u}
\right|^2
=\frac{1}{N}\sum_{u=1}^{N-1}|x_u|^2.
\end{equation}}
\end{remark}
Based on Remark~\ref{Remark1}, \eqref{eq:Pe_DFT_form} simplifies to
\begin{equation}
	P_\mathrm{e}(c_1)
	=\frac{1}{N}\sum_{u=1}^{N-1}\left|x_u\right|^2.
	\label{eq:Pe_parseval}
\end{equation}
Note that the proposed chirp parameter selection can be readily applied to AFDM in the presence of carrier frequency offsets and timing offsets. This is mainly because the proposed schemes are based on the effective channel $\mathbf H_{\rm eff}$ and  the carrier frequency offsets and timing offsets can also be incorporated into $\mathbf H_{\rm eff}$ without changing its banded structure.
Moreover, compared with directly computing $P_\mathrm{A}(c_1)$, calculating $P_\mathrm{e}(c_1)$ from \eqref{eq:Pe_parseval} to characterize the BER performance can significantly reduce the computational complexity.
Therefore, the optimal $c_1$ can be efficiently identified through a finite enumeration over the candidate set using $P_\mathrm{e}(c_1)$ as the metric, which constitutes the proposed FCDA algorithm.

\section{Path-Wise Structured ECM Sparsification for Low-Complexity MMSE Reception}
The dominant computational bottleneck of the MMSE estimator lies in solving the symmetric positive definite (SPD) linear system involved in the estimation process, namely, inverting the regularized matrix $\boldsymbol{\Psi}$ in \eqref{eq:est_signal}. To reduce the computational complexity, we first introduce a sparsified   $\boldsymbol{\Psi}$. Based on this sparsified matrix, we then present the proposed sparsity-exploiting LC-BMMSE receiver, followed by the proposed HS-JCPS for optimal $c_1$ selection.

\begin{figure}[!t]
    \centering
    \subfloat[$\tilde{\boldsymbol{\Psi}}$]{\includegraphics[width=0.15\textwidth]{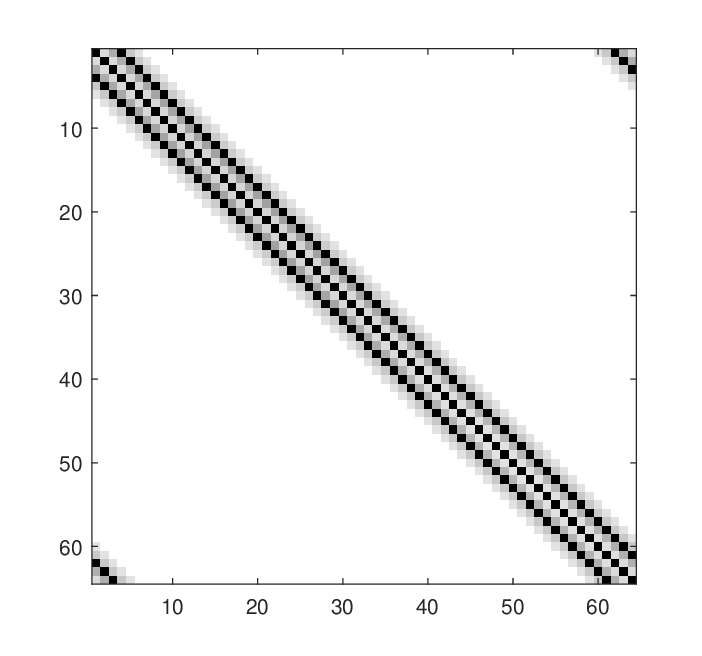}\label{subfig:psi}}\hfill
	\subfloat[$\boldsymbol{\Gamma}$]{\includegraphics[width=0.15\textwidth]{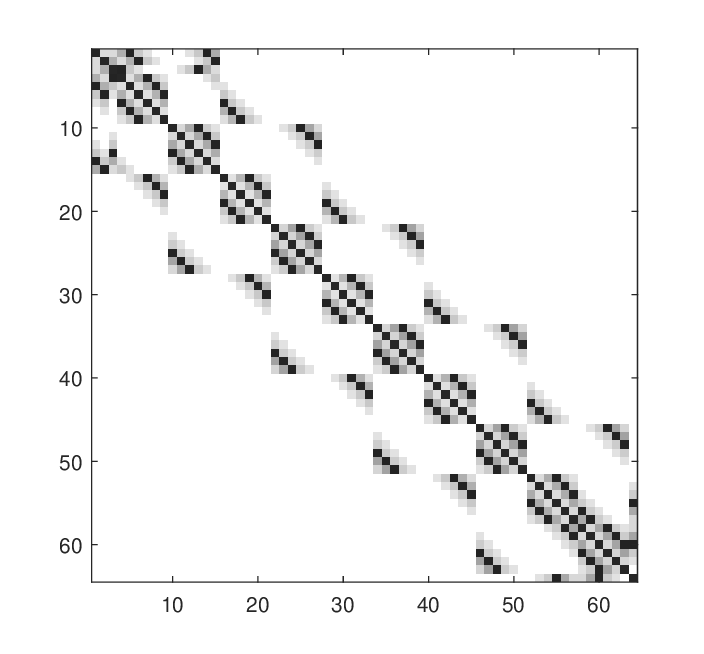}\label{subfig:Gamma}}\hfill
	\subfloat[$\mathbf{R}$]{\includegraphics[width=0.15\textwidth]{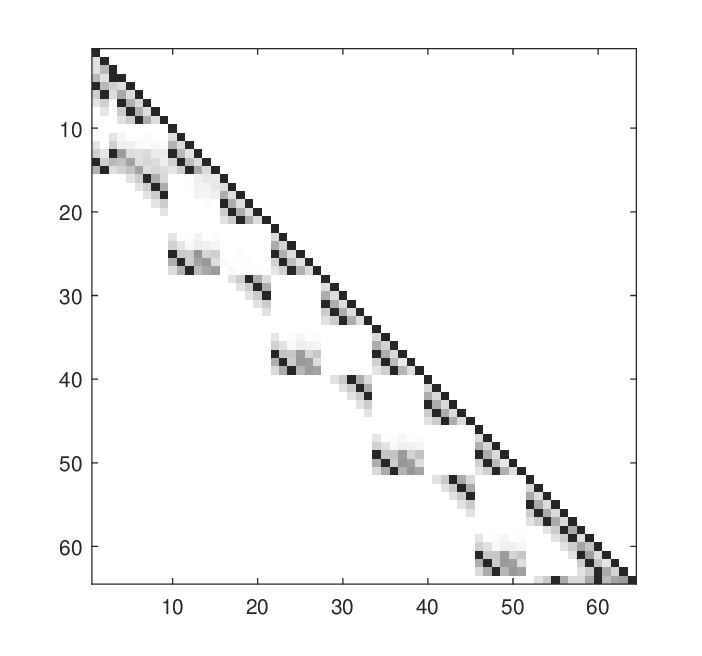}\label{subfig:R}}
	\caption{Structure of $\tilde{\boldsymbol{\Psi}}$, $\boldsymbol{\Gamma}$ and $\mathbf{R}$ with $N = 64$.}
	\label{Fig:RCMStruc}
\end{figure}

\subsection{Sparsification of \texorpdfstring{$\boldsymbol{\Psi}$}{Psi}}\label{subsec:5A}
To characterize the sparse structure of $\boldsymbol{\Psi}$, we introduce a cyclic-banded binary mask matrix corresponding to the DAFT-domain truncation width $d$, given by
\begin{equation}
\mathbf{S}\triangleq \sum_{u=-d+1}^{d-1}\boldsymbol{\Pi}^{u},
\label{eq:S_def}
\end{equation}
where the number of non-zero entries among rows of $\mathbf{S}$ is $\tilde{d}=2d-1$.
Using this mask, the sparsified regularized matrix $\tilde{\boldsymbol{\Psi}}$ is obtained, given by
\begin{equation}
\tilde{\boldsymbol{\Psi}}=\boldsymbol{\Psi}\odot \mathbf{S},
\label{eq:Psi_mask}
\end{equation}
which indicates that only its cyclic-band region specified by $\mathbf{S}$ needs to be considered when calculating $\tilde{\boldsymbol{\Psi}}$.

From \eqref{eq:S_def} and \eqref{eq:Psi_mask}, $\tilde{\boldsymbol{\Psi}}$ exhibits a fixed cyclic-banded sparse structure under the original indexing.
To exploit this structure for low-complexity matrix inversion, we apply the 
reverse Cuthill-McKee (RCM) algorithm \cite{R9,P5} to compute a permutation 
matrix  that transforms the cyclic-banded pattern into a  banded form. 
The reordered matrix is then defined as
\begin{equation}
    \boldsymbol{\Gamma} = \mathbf{P}\tilde{\boldsymbol{\Psi}}\mathbf{P}^\mathcal{T},
\end{equation} 
where  $\mathbf{P}$  is the permutation 
matrix obtained by applying the RCM algorithm to $\mathbf S$. An example of $\tilde{\boldsymbol{\Psi}}$ and the corresponding reordered 
matrix $\boldsymbol{\Gamma}$ is illustrated in Fig.~\ref{Fig:RCMStruc}. 
The permutation preserves the values of all nonzero entries while 
rearranging their positions according to the sparsity pattern of $\mathbf{S}$.
According to \cite{R9}, the reordered matrix yields a bandwidth no greater than $2\tilde{d}$, i.e., 
\begin{equation}\label{eq:Gamma}
	\left[\boldsymbol{\Gamma}\right]_{i,j}=0,\quad \forall\, i,j\ \text{s.t.}\ |i-j|>\tilde d .
\end{equation}

Neglecting the noise term, and using the permutation matrix $\mathbf{P}$ and the reordered matrix $\boldsymbol{\Gamma}$, the estimated symbol vector in \eqref{eq:est_signal} can be rewritten as 
\begin{equation}\label{eq:hat_s_comp}
	\hat{\mathbf{s}}=\mathbf{P}^\mathcal{T} \underbrace{\boldsymbol{\Gamma}^{-1}\overbrace{\mathbf{P}\tilde{\mathbf{y}}}^{\bar{\mathbf{y}}}}_{\bar{\mathbf{s}}},
\end{equation}
where $\tilde{\mathbf{y}}=\tilde{\mathbf{H}}_{\mathrm{eff}}^\mathcal{H}\mathbf{y}$ indicates the output signal of matched filter.
\begin{remark} \label{Remark:hat_s_comp}
\textit{
	From \eqref{eq:hat_s_comp},  computing $\hat{\mathbf{s}}$ involves three steps: matched filtering, i.e.,  $\tilde{\mathbf{y}}=\tilde{\mathbf{H}}_{\mathrm{eff}}^\mathcal{H}\mathbf{y}$, RCM-based permutation, i.e.,  $\bar{\mathbf{y}}=\mathbf{P}\tilde{\mathbf{y}}$ and $\hat{\mathbf{s}}=\mathbf{P}^\mathcal{T}\bar{\mathbf{s}}$, and matrix inversion i.e., $\bar{\mathbf{s}}=\boldsymbol{\Gamma}^{-1}\bar{\mathbf{y}}$. 
	Directly computing $\boldsymbol{\Gamma}^{-1}$ is unnecessary and computationally prohibitive for large $N$. 
	In fact, this inversion process can be equivalent to solving a linear system, i.e., $\boldsymbol{\Gamma}\bar{\mathbf{s}}=\bar{\mathbf{y}}$. }
\end{remark}

\begin{remark} \label{Remark3}
\textit{It is worth noting that $\mathbf{P}$ can be precomputed offline via RCM 
and reused across all channel realizations, since the RCM ordering depends 
only on the sparsity pattern of $\mathbf{S}$, which is uniquely determined 
by $d$ and is independent of the instantaneous channel realization. 
Consequently, $\boldsymbol{\Gamma} = \mathbf{P}\tilde{\boldsymbol{\Psi}}\mathbf{P}^\mathcal{T}$ 
exhibits a fixed banded structure with bandwidth fully determined by 
$d$, enabling low-complexity receiver implementation as discussed in the 
following subsection.
}
\end{remark}

\subsection{Implementation of the Sparsity-Exploiting LC-BMMSE Receiver}
\label{subsec:BCD}
This subsection details the low-complexity implementation of $\hat{\mathbf{s}}$ in Remark~\ref{Remark:hat_s_comp}, including: matched filtering $\tilde{\mathbf{y}}=\tilde{\mathbf{H}}_{\mathrm{eff}}^\mathcal{H}\mathbf{y}$, RCM-based permutation $\bar{\mathbf{y}}=\mathbf{P}\tilde{\mathbf{y}}$, and $\hat{\mathbf{s}}=\mathbf{P}^\mathcal{T}\bar{\mathbf{s}}$, and the banded Cholesky-based solution of the linear system $\boldsymbol{\Gamma}\bar{\mathbf{s}}=\bar{\mathbf{y}}$.
\subsubsection{Matched Filtering $\tilde{\mathbf{y}}=\tilde{\mathbf{H}}_{\mathrm{eff}}^\mathcal{H}\mathbf{y}$ }
Owing to the sparsified structure of the ECM, the matched filtering operation can be carried out by exploiting its sparse pattern, thereby avoiding full-matrix computation.

\subsubsection{RCM-Based Permutation \texorpdfstring{$\mathbf{P}\tilde{\mathbf{y}}$}{Py} 
and \texorpdfstring{$\mathbf{P}^\mathcal{T}\bar{\mathbf{s}}$}{PTs}}\label{sssec:Reordering}
Since $\mathbf{P}$ is a permutation matrix, both $\bar{\mathbf{y}} = 
\mathbf{P}\tilde{\mathbf{y}}$ and $\hat{\mathbf{s}} = \mathbf{P}^\mathcal{T}
\bar{\mathbf{s}}$ reduce to index  mapping operations, which  can be implemented with negligible overhead.

\subsubsection{Solving  \texorpdfstring{$\boldsymbol{\Gamma}\bar{\mathbf{s}}=\bar{\mathbf{y}}$}{Gamma}}\label{sssec:Bandsov}
We exploit the banded structure of $\boldsymbol{\Gamma}$ and solve the linear system via a banded Cholesky factorization followed by forward/backward substitutions \cite{P6}.
By performing the band Cholesky factorization of $\boldsymbol{\Gamma}$, one has $\boldsymbol{\Gamma}=\mathbf{R}\mathbf{R}^\mathcal{H}$, where $\mathbf{R}$ is a lower triangular matrix with lower bandwidth $\tilde{d}$.
The structure of $\mathbf{R}$ is shown in Fig.~\ref{subfig:R}.
Specifically, for the $j$th column, only the entries within $\left[j:\min(j+\tilde{d},N)\right]$ are involved in the factorization.
The update of the $j$th column depends only on the previous columns that overlap with this band, namely, $k=\max(1,j-\tilde{d}),\cdots,j-1$.
For each such $k$, the overlapping interval is performed, i.e.,
\begin{equation}
    [\boldsymbol{\Gamma}]_\mathsmaller{j:\lambda_k,j}=[\boldsymbol{\Gamma}]_\mathsmaller{j:\lambda_k,j}-[\boldsymbol{\Gamma}]_\mathsmaller{j,k}\,[\boldsymbol{\Gamma}]_\mathsmaller{j:\lambda_k,k},
\end{equation}
where $\lambda_k=\min(k+\tilde{d},N)$.
After all overlapping contributions are accumulated, the resulting band segment is normalized by $\sqrt{[\boldsymbol{\Gamma}]_\mathsmaller{j,j}}$ to obtain the $j$th column of the Cholesky factor $\mathbf{R}$.
After obtaining $\mathbf{R}$, the linear system can be re-expressed as
\begin{equation}\label{eq:Chol_linear}
	\bar{\mathbf{y}} = \mathbf{R}\underbrace{\mathbf{R}^\mathcal{H}\bar{\mathbf{s}}}_{\tilde{\mathbf{s}}}.
\end{equation}
Based on \eqref{eq:Chol_linear},  $\bar{\mathbf{s}}$ is recovered by solving two banded triangular 
systems sequentially.  
 First, the lower triangular system 
$\mathbf{R}\tilde{\mathbf{s}} = \bar{\mathbf{y}}$ is solved by forward 
substitution.
Exploiting the banded structure of $\mathbf{R}$, the $i$th component of $\tilde{\mathbf{s}}$ is computed as
\begin{equation}
	[\tilde{\mathbf{s}}]_i = \frac{1}{\left[\mathbf{R}\right]_{i,i}}\left([\bar{\mathbf{y}}]_i - \sum_{\mathsmaller{j=\max(1,\,i-\tilde{d})}}^{\mathsmaller{i-1}}  \left[\mathbf{R}\right]_{i,j}\,[\tilde{\mathbf{s}}]_j\right).
\end{equation}
Then, the upper triangular system $\mathbf{R}^\mathcal{H}\bar{\mathbf{s}} = 
\tilde{\mathbf{s}}$ is solved by backward substitution. 
The $i$th component of $\bar{\mathbf{s}}$ is computed as
\begin{equation}
	[\bar{\mathbf{s}}]_i	= \frac{1}{\left[\mathbf{R}\right]_{i,i}}\left(
	[\tilde{\mathbf{s}}]_i - \sum_{\mathsmaller{j=i+1}}^{\mathsmaller{\min(i+\tilde{d},N)}} \left[\mathbf{R}\right]_{j,i}\,[\bar{\mathbf{s}}]_j\right).
\end{equation}
Both substitutions are restricted to the band, so each row involves at most $\tilde{d}$ operations.
Finally, the estimated vector  is obtained by
\begin{equation}
	\hat{\mathbf{s}} = \mathbf{P}^\mathcal{T} \bar{\mathbf{s}}.
\end{equation}
Overall, the detailed steps of the proposed sparsity-exploiting LC-BMMSE receiver are presented in Algorithm \ref{alg:Low-Complexity Equalizer}.

\begin{algorithm}[t]
	\caption{The proposed sparsity-exploiting LC-BMMSE receiver}
	\label{alg:Low-Complexity Equalizer}
	\renewcommand{\algorithmicrequire}{\textbf{Input:}}
	\renewcommand{\algorithmicensure}{\textbf{Output:}}
	\begin{algorithmic}[1]
			\REQUIRE{Vector $\mathbf{y} \in \mathbb{C}^{N \times 1}$, Matrix $\mathbf{S}$, and $\frac{\sigma_\mathrm{n}^{2}+\sigma_\Delta^{2}}{\mathrm{E}_\mathrm{s}} \in \mathbb{R}^{+}$}
			\ENSURE{Estimated vector $\hat{\mathbf{s}} \in \mathbb{C}^{N \times 1}$}
			\STATE Compute $\tilde{\mathbf{y}}$;
			\STATE Use the precomputed permutation matrix $\mathbf P$ obtained from the RCM ordering of $\mathbf S$;
			\STATE Reorder $\tilde{\boldsymbol{\Psi}}$: $\boldsymbol{\Gamma}= \mathbf{P} \tilde{\boldsymbol{\Psi}} \mathbf{P}^\mathcal{T}$;
			\STATE Rearrange $\tilde{\mathbf{y}}$ through the permutation matrix $\mathbf{P}$ to obtain $\bar{\mathbf{y}}$;
			\STATE Compute $\mathbf{R}$ using band Cholesky factorization;
			\STATE Compute	$\bar{\mathbf{y}}=\mathbf{R}\tilde{\mathbf{s}}$ and $\tilde{\mathbf{s}}=\mathbf{R}^\mathcal{H}\bar{\mathbf{s}}$ using low-complexity forward and backward substitution respectively;
			\STATE Rearrange $\bar{\mathbf{s}}$ through the permutation matrix $\mathbf{P}^\mathcal{T}$ to obtain $\hat{\mathbf{s}}$, i.e., $\hat{\mathbf{s}} = \mathbf{P}^\mathcal{T} \bar{\mathbf{s}}$;
			\RETURN $\hat{\mathbf{s}}$
	\end{algorithmic}
\end{algorithm}

\subsection{The Proposed HS-JCPS}\label{sec:SparsH}
Having prescribed the sparsity mask matrix $\mathbf{S}$ for $\tilde{\boldsymbol{\Psi}}$, 
we next construct a sparsified ECM $\tilde{\mathbf{H}}_{\mathrm{eff}}$ referring to \eqref{eq:sparECM} such that the induced 
$\tilde{\boldsymbol{\Psi}}$ complies with \eqref{eq:Psi_mask}.
For convenience, we define the pruning-width vector as $\mathbf{w}$, which contains the pruning width $w_i$ for all paths.
Nevertheless, such sparsification inevitably removes part of the useful channel components and introduces a model mismatch relative to the original ECM, which generally leads to SINR degradation in MMSE detection. 
To quantify the performance loss caused by sparsification, we establish a lower bound on the achievable SINR of the resulting MMSE estimator.

\begin{lemma}\label{lemma1}
	The MMSE estimator using the sparsified ECM $\tilde{\mathbf{H}}_\mathrm{eff}$ is expressed as $\tilde{\mathbf{W}}=\tilde{\boldsymbol{\Psi}}^{-1}\tilde{\mathbf{H}}_\mathrm{eff}^\mathcal{H}$.
	A lower bound of SINR for MMSE detection based on $\tilde{\mathbf{W}}$ satisfies
	\begin{equation}
		\operatorname{SINR}_k \geq \frac{\left[\tilde{\mathbf{M}}\right]_{k,k}}{1+\frac{\left[\tilde{\mathbf{M}}\mathbf{D}^\mathcal{H}\right]_{k,k}}{\left[\tilde{\mathbf{M}}\right]_{k,k}}-\left[\tilde{\mathbf{M}}\right]_{k,k}},
		\label{eq:SINRbound}
	\end{equation}
	where $\tilde{\mathbf{M}}=\tilde{\boldsymbol{\Psi}}^{-1}\tilde{\mathbf{G}}+\mathbf{D}$ indicates the equivalent matrix induced by the MMSE estimator  and $\mathbf{D} = \tilde{\boldsymbol{\Psi}}^{-1}\tilde{\mathbf{H}}_\mathrm{eff}^\mathcal{H}\Delta\mathbf{H}_\mathrm{eff}$ represents the interference matrix.
\end{lemma}
\begin{proof}
	For a detailed derivation, refer to Appendix~\ref{appendixA}.
\end{proof}


It should be noted that Lemma~\ref{lemma1} is used to characterize the SINR degradation induced by ECM sparsification. The lower bound serves as an analytical reference for quantifying the performance degradation caused by the sparsification perturbation $\Delta\mathbf H_{\rm eff}$.
Specifically, the discarded channel component induces the interference term $\mathbf{D}$, which degrades the effective SINR. 
Hence, a smaller $\Delta\mathbf{H}_\mathrm{eff}$ leads to a tighter SINR lower bound.

Beyond the SINR degradation characterized in Lemma~\ref{lemma1}, ECM sparsification can also affect the diversity behavior as it changes the path-induced error responses retained by the receiver.

\begin{remark}
\textit{Following the rank-based AFDM diversity criterion, let $\boldsymbol{\delta}$ denote a nonzero symbol-error vector, and define the sparse path-induced error-response matrix as
\begin{equation}
	\boldsymbol{\Phi}_{\rm sp}(\boldsymbol{\delta};\mathbf w)
=
[\tilde{\mathbf H}_{{\rm eff},1}\boldsymbol{\delta},\ldots,
\tilde{\mathbf H}_{{\rm eff},P}\boldsymbol{\delta}].
\end{equation}
The diversity order supported by the sparsified ECM can be characterized as
\begin{equation}
	\rho_{\rm sp}=
\min_{\boldsymbol{\delta}\neq\mathbf 0}
{\rm rank}
\left(
\boldsymbol{\Phi}_{\rm sp}(\boldsymbol{\delta};\mathbf w)
\right)
\le P.
\end{equation}
Thus, the sparsified ECM preserves the full path diversity as long as the retained path-induced responses keep full column rank, i.e.,
${\rm rank}(\boldsymbol{\Phi}_{\rm sp}(\boldsymbol{\delta};\mathbf w))=P$
for all nonzero error vectors. This condition can be satisfied when the pruning widths retain the dominant coupling neighborhood of each physical path. In the proposed path-wise sparsification, the retained support is centered at the dominant DAFT-domain coupling locations, and the widths are jointly optimized with $c_1$ to preserve the main path-induced components. Therefore, the discarded part mainly contributes to the sparsification perturbation $\Delta\mathbf H_{\rm eff}$ and the corresponding residual interference, rather than necessarily reducing the diversity order.}
\end{remark}

For a fixed $\mathbf{S}$, multiple feasible sparsified ECMs 
$\tilde{\mathbf{H}}_{\mathrm{eff}}$ may exist, corresponding to different 
pruning-width vector $\mathbf{w}$ that induce the same cyclic-banded support 
of $\tilde{\boldsymbol{\Psi}}$. For example, under $c_1 = 1/(2N)$, the 
pruning-width vectors $\mathbf{w} \in \{\{0,2\}, \{1,1\}, \{2,0\}\}$ yield 
the sparsified ECM structures $\tilde{\mathbf{H}}_{\mathrm{eff}}$ shown in Figs.~\ref{fig1}--\ref{fig3}, with 
the corresponding $\tilde{\boldsymbol{\Psi}}$ structures shown in 
Figs.~\ref{fig4}--\ref{fig6}. As long as the maximum number of nonzero 
entries per row of the sparsified ECM equals $d = 6$, the induced 
$\tilde{\boldsymbol{\Psi}}$ shares an identical cyclic-banded support with 
at most $\tilde{d} = 11$ nonzero entries per row. Consequently, all such 
candidates incur identical solver complexity for 
$\boldsymbol{\Gamma}\bar{\mathbf{s}} = \bar{\mathbf{y}}$, yet their 
detection performance may be different, motivating the MMSE-optimal 
sparsification design developed in the rest of this section.
 
\begin{figure*}[!t]
    \centering
	\subfloat[\small{${\tilde{\mathbf{H}}_\mathrm{eff}}, \mathbf{w}=\left\{0,2\right\}$}]{\includegraphics[width=0.15\textwidth]{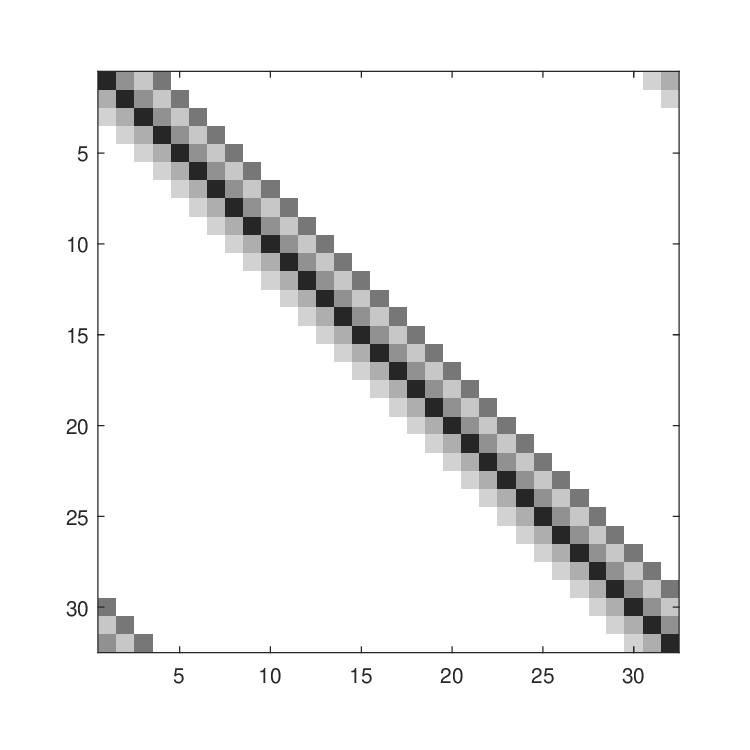}\label{fig1}}\hfill
	\subfloat[\small{${\tilde{\mathbf{H}}_\mathrm{eff}},\mathbf{w}=\left\{1,1\right\}$}]{\includegraphics[width=0.151\textwidth]{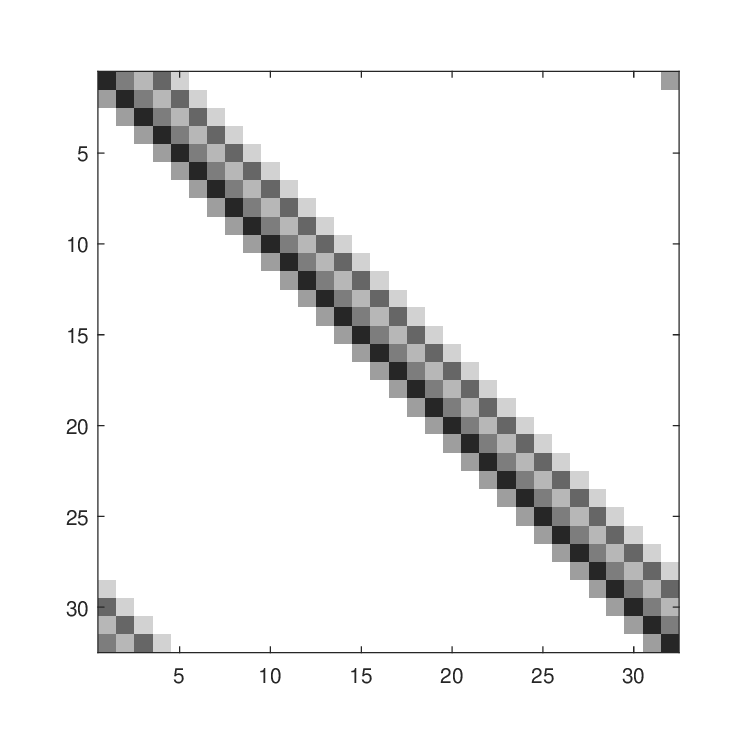}\label{fig2}}\hfill
	\subfloat[\small{${\tilde{\mathbf{H}}_\mathrm{eff}},\mathbf{w}=\left\{2,0\right\}$}]{\includegraphics[width=0.15\textwidth]{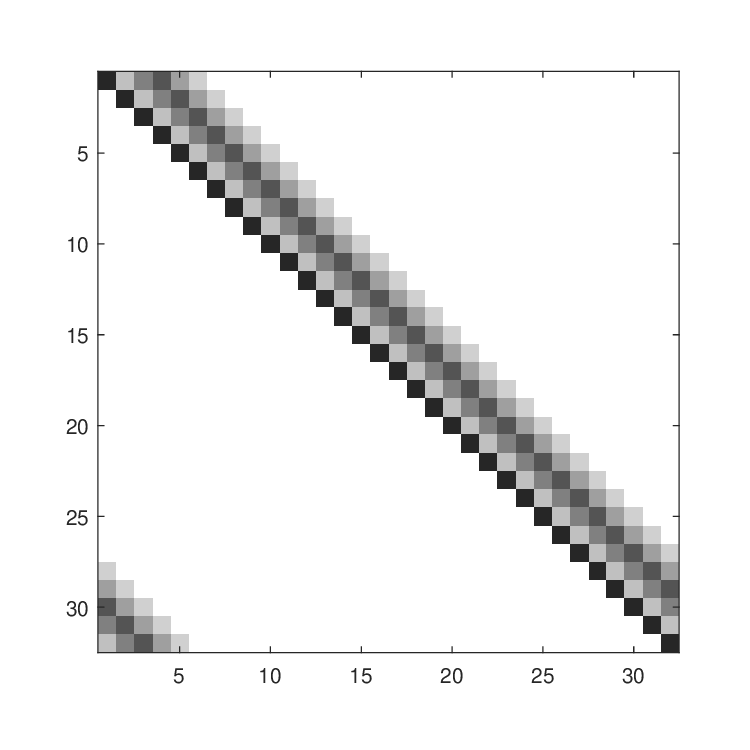}\label{fig3}}\hfill
	\subfloat[\small{${\tilde{\boldsymbol{\Psi}}},\mathbf{w}=\left\{0,2\right\}$}]{\includegraphics[width=0.15\textwidth]{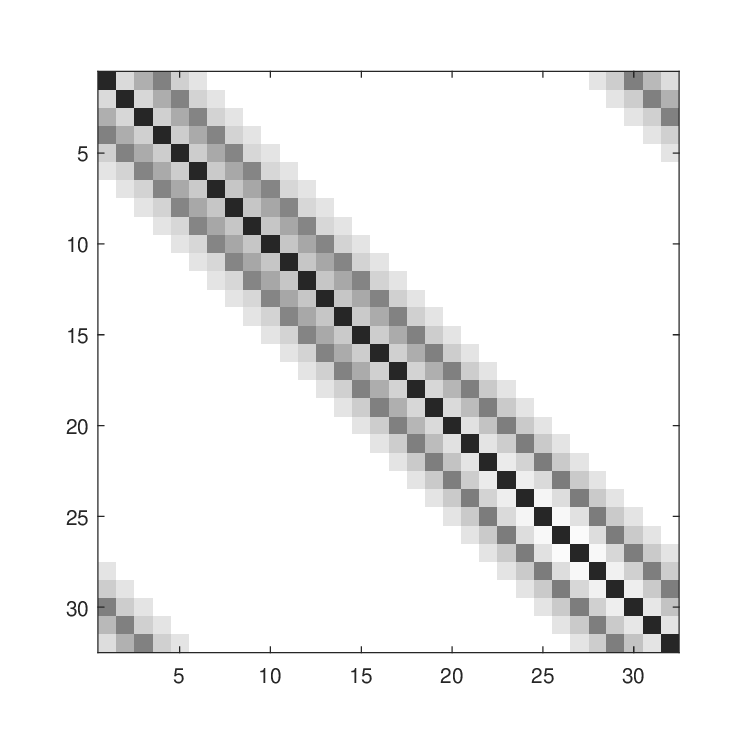}\label{fig4}}\hfill
    \subfloat[\small{${\tilde{\boldsymbol{\Psi}}},\mathbf{w}=\left\{1,1\right\}$}]{\includegraphics[width=0.15\textwidth]{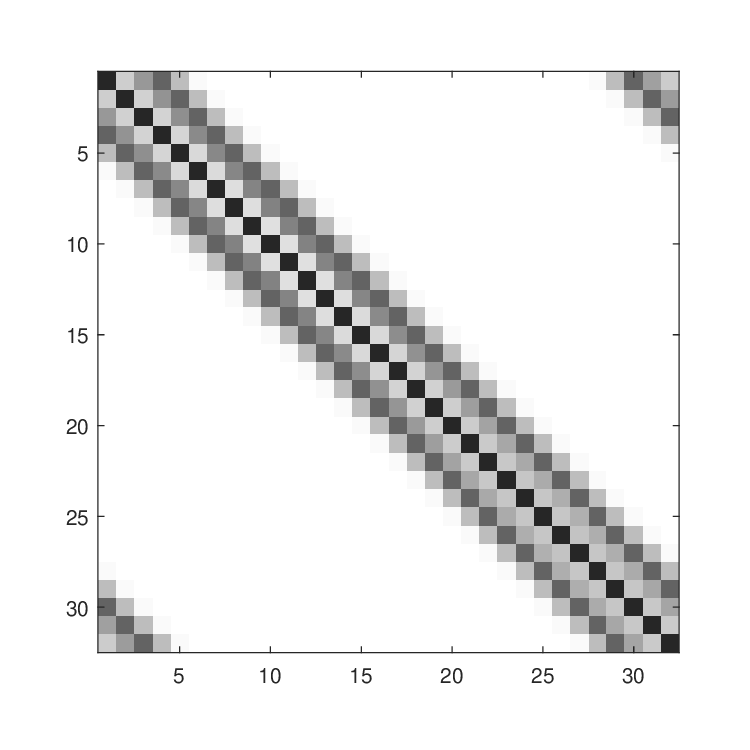}\label{fig5}}\hfill
	\subfloat[\small{${\tilde{\boldsymbol{\Psi}}},\mathbf{w}=\left\{2,0\right\}$}]{\includegraphics[width=0.15\textwidth]{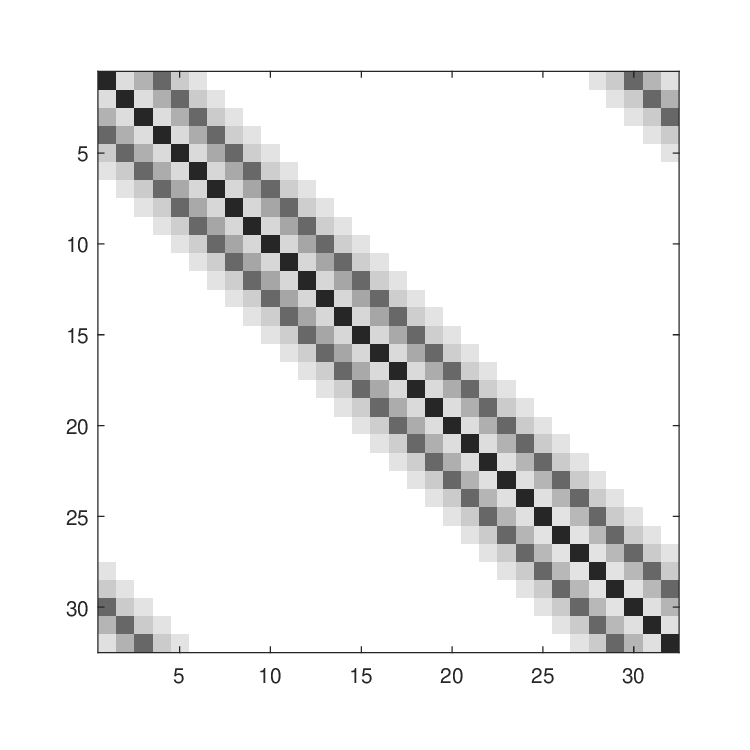}\label{fig6}}\hfill
	\caption{Structures of $\tilde{\mathbf{H}}_\mathrm{eff}$ and $\tilde{\boldsymbol{\Psi}}$ with normalized time delay vectors $\left\{0,1\right\}$, $c_1=1/{2N}$ and $N = 32$. \protect\subref{fig1}, \protect\subref{fig2} and \protect\subref{fig3} have different structures, but the maximum number of entries among rows is $d=6$. \protect\subref{fig4}, \protect\subref{fig5} and \protect\subref{fig6} have the same matrix structure, and its maximum number of entries among rows is $\tilde{d}=11$.}
	\label{Fig:Matrices_Structure}
\end{figure*}

In the following, we develop a joint optimization of the AFDM parameter $c_1$ and the pruning-width vector $\mathbf{w}$ under a fixed complexity budget. 
Specifically, with the sparsity mask $\mathbf{S}$ fixed by the truncation 
width $d$, the parameter $c_1$ governs the relative peak locations of 
different paths in the DAFT domain, while $\mathbf{w}$ controls the 
retained support around each peak. Together, they determine 
$\tilde{\mathbf{H}}_{\mathrm{eff}}(\mathbf{w})$ and hence the MMSE 
detection performance associated with $\tilde{\boldsymbol{\Psi}}$. Under 
the pruning-width budget $\sum_{i=1}^P w_i = w_\mathrm{s}$, the 
optimization problem is formulated as
\begin{align}
	\begin{aligned}
		&\min_{c_1,\mathbf{w} \in \mathbb{Z}_{+}^P} \quad P_\mathrm{A}\left(c_1,\mathbf{w}\right) \\
		\text { s.t. } 
		& \mathbf{w}= [w_1,\cdots,w_P]^\mathcal{T}\in \mathbb{Z}_+^P,\ \sum_{i=1}^P w_i =w_\mathrm{s},\\
		& c_1 = \frac{k}{2N}, \ k \in \{1,\ldots,N\},\ \tilde{\boldsymbol{\Psi}}=\boldsymbol{\Psi}\odot \mathbf{S},
		\label{eq:Sparse OptP2}
	\end{aligned}
\end{align}
where $P_\mathrm{A}\left(c_1,\mathbf{w}\right)$ is evaluated from the sparsified ECM $\tilde{\mathbf{H}}_\mathrm{eff}(\mathbf{w})$ based on \eqref{eq:SINRk} and \eqref{eq:Calber}. 
In general, minimizing $P_\mathrm{A}(c_1,\mathbf w)$ is also helpful for reducing the condition number, as it tends to reject sparsified ECMs with weak singular directions, thereby improving their conditioning.  In addition, while $P_\mathrm{A}(c_1,\mathbf w)$ is specifically designed for the LC-MMSE detector, our proposed HS-JCPS can also be readily extended to the ZF detector by following a similar approaches of \eqref{eq:Sparse OptP2}.
As mentioned, with the sparsity mask $\mathbf{S}$ fixed, the receiver complexity is 
determined, and the remaining task is to identify the feasible 
pair $(c_1, \mathbf{w})$ that minimizes  $P_\mathrm{A}(c_1, 
\mathbf{w})$ under the same structural constraint.
As $c_1$ takes values from a finite discrete set and each $w_p$ is an integer constrained within a finite interval, an exhaustive search over all $(c_1,\mathbf{w})$ combinations can be employed to obtain the optimal solution. 
However, such an approach becomes intractable and computationally prohibitive when the number of paths $P$ or the width range is large. 

 To address the above issues, we develop the HS-JCPS algorithm, which efficiently 
approximates the optimal solution under reduced search complexity. 
Specifically, the proposed HS-JCPS algorithm consists of two loops. 
The outer loop generates feasible candidates of $c_1$ together with their corresponding initial pruning-width vectors, while the inner loop performs a local heuristic search over $\mathbf{w}$ for each fixed $c_1$.

\subsubsection{The Outer Loop}
For a fixed cyclic band width $d$,  the total pruning-width budget is set as
\begin{equation}\label{eq:ws}
	w_\mathrm{s}=\sum_{i=1}^{P} w_i,
\end{equation}
so that all candidate sparsified ECMs are compared under the same complexity budget. 
For each candidate $c_1$, the peak locations 
$\{\xi_i\}_{i=1}^{P}$ are first determined, and the discrete candidate set of 
$c_1$ is restricted by
\begin{equation}
    2Nc_1 \leq \left\lfloor \frac{d}{l_P - l_1} \right\rfloor = r_{\max} ,
\end{equation}
to ensure all path peaks lie within the prescribed cyclic band. Under this 
constraint, the pruning widths of the first and last paths satisfy
\begin{equation}\label{eq:w1wP}
	w_1+w_P=d-(\xi_P-\xi_1),
\end{equation}
which is determined by the cyclic band width $d$ and the relative peak separation $\xi_P-\xi_1$. 
Accordingly, for each feasible candidate $c_1$, a set of feasible initial pruning-width vectors can be generated by enumerating the admissible values of $w_1$, while $w_P$ is obtained from the above relation. 
Let $r$ denote the index of the candidate $c_1$, and let $l$ denote the index of the corresponding feasible initialization. 
Then, the resulting initial pruning-width vector is denoted by $\mathbf{w}_{r,l}^{(0)}$.

\subsubsection{The Inner Loop}
For each feasible pair $(r,l)$, the inner loop refines the corresponding pruning-width vector by a local search under the fixed budget $w_\mathrm{s}$. 
Starting from the initialization $\mathbf{w}_{r,l}^{(0)}$, the algorithm updates the current solution iteratively. 
To avoid convergence to a local minimum, the local search can be conducted from multiple feasible initializations, thereby exploring different feasible regions of the pruning-width space. 
At iteration $t$, let $\mathbf{w}_{r,l}^{(t)}$ denote the current pruning-width vector. 
{The feasible neighboring candidates are generated by transferring the width step size $\Delta w$ between two paths while keeping the total budget unchanged.
The two paths are selected by enumerating all feasible path pairs.
Let $\mathbf e_i,\mathbf e_j\in\mathbb Z^P_+$ denote the unit vectors of path $i$ and path $j$, respectively, each having a single nonzero entry equal to one at its path index.
For each feasible pair, the neighboring candidate is generated as
$\mathbf w_{r,l}^{\rm cand}=\mathbf w_{r,l}^{(t)}+\Delta w\mathbf e_i-\Delta w\mathbf e_j.$}
Based on $\mathbf{w}_{r,l}^{\mathrm{cand}}$, the sparsified ECM 
$\tilde{\mathbf{H}}_{\mathrm{eff}}(\mathbf{w}_{r,l}^{\mathrm{cand}})$ is 
constructed via \eqref{eq:sparECM}, and the corresponding cost 
$P_\mathrm{A}(c_1,\mathbf{w}_{r,l}^{\mathrm{cand}})$ is evaluated. 
If the candidate yields a lower cost than the current solution, it is accepted and the current solution is updated accordingly; otherwise, the 
current solution remains unchanged. Namely, the update rule is given by
\begin{equation}
    \mathbf{w}_{r,l}^{(t+1)}=
    \begin{cases}
        \mathbf{w}_{r,l}^{\mathrm{cand}}, & \text{if } 
        P_\mathrm{A}(c_1, \mathbf{w}_{r,l}^{\mathrm{cand}}) < 
        P_\mathrm{A}(c_1, \mathbf{w}_{r,l}^{(t)}), \\
        \mathbf{w}_{r,l}^{(t)}, & \text{otherwise.}
    \end{cases}
\end{equation}
The iteration continues until no improving feasible neighbor exists or the 
maximum number of iterations is reached.

After all pairs $(r,l)$ have been explored, the final solution is selected as
\begin{equation}
    \left(c_1^\star, \mathbf{w}^\star\right) = 
    \arg\min_{(r,l)}\, P_\mathrm{A}\!\left(c_{1,r}, 
    \mathbf{w}_{r,l}^{\star}\right),
\end{equation}
where $\mathbf{w}_{r,l}^{\star}$ is the locally optimized vector for the 
$(r,l)$-th initialization. 

{It should be emphasized that $P_\mathrm{A}(c_1,\mathbf w)$ is a discrete nonconvex objective with respect to $\mathbf w$. Therefore, the inner-loop local search does not claim global optimality for a single initialization. Instead, HS-JCPS employs a multi-start local-search structure: for each candidate $c_1$, multiple feasible initial pruning-width vectors ${\mathbf w_{r,l}^{(0)}}$ are generated in the outer loop, and the inner local search is independently launched from each initialization. The final solution is selected as the best locally optimized result among all explored candidate $c_1$ values and feasible initializations.}
The proposed HS-JCPS algorithm is summarized in 
Algorithm~\ref{alg:opt_sparse_ecm}, which efficiently approximates the 
optimal pair $(c_1, \mathbf{w})$ without resorting to exhaustive search.

\begin{algorithm}[t]
\caption{HS-JCPS Strategy}
\label{alg:opt_sparse_ecm}
\renewcommand{\algorithmicrequire}{\textbf{Input:}}
\renewcommand{\algorithmicensure}{\textbf{Output:}}
\begin{algorithmic}[1]
\REQUIRE Channel matrices $\{\mathbf{H}_i\}_{i=1}^{P}$, band width $d$, maximum iteration number $\mathrm{Iter}_{\max}$, step size $\Delta w$;
\ENSURE Near-optimal pair $\left(c_1^\star,\mathbf{w}^\star\right)$, sparsified ECM $\widetilde{\mathbf{H}}_{\mathrm{eff}}^\star$, and minimum cost $P_\mathrm{A}^\star$;

\STATE Initialize $P_\mathrm{A}^\star \leftarrow +\infty$, $\mathbf{w}_{r,l}^{(0)}$;
\FOR{each candidate $c_{1,r}$ satisfying $2Nc_{1,r}\le r_{\max}$}
    \STATE Calculate the peak locations $\{\xi_i\}_{i=1}^{P}$;
    \FOR{each feasible initialization indexed by $l$}
        \STATE Set $t\leftarrow 0$;
        \REPEAT
            \STATE Set $\mathbf{w}_{r,l}^{(t+1)} \leftarrow \mathbf{w}_{r,l}^{(t)}$, $\mathrm{flag}\leftarrow 0$;
            \FOR{each feasible neighboring candidate $\mathbf{w}_{r,l}^{\mathrm{cand}}$}
                \STATE Construct $\widetilde{\mathbf{H}}_{\mathrm{eff}}\!\left(\mathbf{w}_{r,l}^{\mathrm{cand}}\right)$ according to \eqref{eq:sparECM};
                \STATE Evaluate $P_\mathrm{A}\!\left(c_{1,r},\mathbf{w}_{r,l}^{\mathrm{cand}}\right)$ using \eqref{eq:SINRk} and \eqref{eq:Calber};
                \IF{$P_\mathrm{A}\!\left(c_{1,r},\mathbf{w}_{r,l}^{\mathrm{cand}}\right) < P_\mathrm{A}\!\left(c_{1,r},\mathbf{w}_{r,l}^{(t)}\right)$}
                    \STATE $\mathbf{w}_{r,l}^{(t+1)} \leftarrow \mathbf{w}_{r,l}^{\mathrm{cand}}$, $\mathrm{flag}\leftarrow 1$;
                \ENDIF
            \ENDFOR
            \STATE $t\leftarrow t+1$;
        \UNTIL{$\mathrm{flag}=0$ \OR $t=\mathrm{Iter}_{\max}$}
        \STATE Set $\mathbf{w}_{r,l}^{\star}\leftarrow \mathbf{w}_{r,l}^{(t)}$;
        \IF{$P_\mathrm{A}\!\left(c_{1,r},\mathbf{w}_{r,l}^{\star}\right)<P_\mathrm{A}^\star$}
            \STATE $P_\mathrm{A}^\star \leftarrow P_\mathrm{A}\!\left(c_{1,r},\mathbf{w}_{r,l}^{\star}\right)$, $c_1^\star \leftarrow c_{1,r}$, $\mathbf{w}^\star \leftarrow \mathbf{w}_{r,l}^{\star}$;
            \STATE $\widetilde{\mathbf{H}}_{\mathrm{eff}}^\star \leftarrow \widetilde{\mathbf{H}}_{\mathrm{eff}}\!\left(\mathbf{w}_{r,l}^{\star}\right)$;
        \ENDIF
    \ENDFOR
\ENDFOR
\RETURN $c_1^\star,\mathbf{w}^\star,\widetilde{\mathbf{H}}_{\mathrm{eff}}^\star,P_\mathrm{A}^\star$
\end{algorithmic}
\end{algorithm}

\begin{table*}[t]
	\centering
	\caption{Online Complexity Analysis of Sparsity-Exploiting LC-BMMSE Receiver}
	\label{tab:ComplexityAnaly}
	\begin{tabular}{|c|c|c|c|c|}
		\hline
		Algorithm & $\tilde{\boldsymbol{\Psi}}^{-1}\bar{\mathbf{y}}$ & $\mathbf{P}^\mathcal{T}\tilde{\mathbf{y}}$ and $\mathbf{P} \bar{\mathbf{s}}$ &$\tilde{\mathbf{H}}_{\mathrm{eff}}^\mathcal{H}\mathbf{y}$ &  Online Total Number of Complex Multiplications \\
		\hline
		Conventional MMSE & $\frac{4}{3}N^3$ &  $/$ &$N^2$ & $\frac{4}{3}N^3+N^2$  \\
		\hline
		LC-BMMSE & $N {\tilde{d}}^2 + 2 N {\tilde{d}}$ & $0$ &$Nd$ & $N {\tilde{d}}^2 + 2 N {\tilde{d}}+Nd$ \\
		\hline
	\end{tabular}
\end{table*}

\subsection{Complexity Analysis}
In this subsection, we analyze the computational complexity of the proposed sparsity-exploiting LC-BMMSE receiver and the proposed HS-JCPS algorithm.

\paragraph{The Complexity of Sparsity-Exploiting LC-BMMSE Receiver}
{The computational cost includes offline preprocessing and online receiver processing. The RCM-based $\mathbf P$ is obtained offline from the sparsity pattern of $\mathbf S$, independent of instantaneous channel coefficients, and can be reused under the same sparsity pattern.
Since the corresponding graph contains $N$ vertices and at most $O(N\tilde{d})$ adjacency entries, the RCM procedure requires $O(N\tilde{d})$ graph-traversal and index operations, while involving no complex multiplications. 
We next analyze the online receiver-processing complexity.}
Since $\mathbf{P}$ is a permutation matrix, the operations $\mathbf{P}\tilde{\mathbf{y}}$ and $\mathbf{P}^\mathcal{T}\bar{\mathbf{s}}$ involve only index reordering and do not require complex multiplications. 
Due to the sparsified ECM, computing $\tilde{\mathbf{H}}_{\mathrm{eff}}^\mathcal{H}\mathbf{y}$ requires only $Nd$ complex multiplications, instead of $N^2$. 
As shown in \eqref{eq:Gamma}, $\boldsymbol{\Gamma}$ is banded with a maximum bandwidth of $2\tilde{d}$, and its in-place banded Cholesky factorization $\boldsymbol{\Gamma}=\mathbf{R}\mathbf{R}^\mathcal{H}$ requires approximately $N\tilde{d}^{\,2}$ complex multiplications. 
Then, the forward and backward substitution requires at most $2N\tilde{d}$ complex multiplications, since each row of $\mathbf{R}$ contains no more than $\tilde{d}$ nonzero off-diagonal elements. 
Therefore, the online total complexity of the proposed sparsity-exploiting LC-BMMSE receiver is upper-bounded by $N\tilde{d}^{\,2}+2N\tilde{d}+Nd$.
By contrast, a conventional full-matrix MMSE receiver based on the dense ECM requires either dense Cholesky factorization or explicit matrix inversion, both incurring approximately $\frac{4}{3}N^3$ complex multiplications, plus an additional $N^2$ cost for matrix multiplication. 
The resulting complexity comparison is summarized in Table~\ref{tab:ComplexityAnaly}.

\paragraph{The Complexity of HS-JCPS}
Let $\mathrm{Iter}_{\max}$ and $C_\mathrm{eval}$ denote the number of local-search iterations and the cost of constructing the sparsified ECM and evaluating $P_\mathrm{A}(c_1,\mathbf{w})$, respectively.
For a given $c_{1,r}$ in the discrete candidate set, the number of feasible initializations can be approximated by $L_r\approx \left\lfloor \frac{d-(\xi_P-\xi_1)}{\Delta w}\right\rfloor+1$.
In each local-search iteration, a neighboring candidate is generated by increasing the width of one path by $\Delta w$ and decreasing that of another path, while keeping the total width budget unchanged. 
Hence, the number of feasible neighboring candidates is on the order of $\mathcal{O}(P^2)$. 
Therefore, the complexity of HS-JCPS can be approximated as $\mathcal{O}\!\left(r_{\max}\frac{d}{\Delta w}\mathrm{Iter}_{\max}P^2C_\mathrm{eval}\right)$.
By contrast, exhaustive search requires evaluating all feasible width allocations under the budget constraint in \eqref{eq:ws}. 
Let $Q=\frac{w_\mathrm{s}}{\Delta w}$ denote the total number of discrete width units to be allocated among the $P$ paths.
Then, its complexity is $\mathcal O\!\left( r_{\max}\binom{Q+P-1}{P-1}C_{\mathrm{eval}}    \right)$.
This shows that the exhaustive search grows combinatorially with the path number $P$ and the width budget $Q$, whereas the proposed HS-JCPS replaces the combinatorial search over $\mathbf{w}$ with a local search initialized from a limited set of feasible starting points.

\begin{table}
	\centering
	\caption{AFDM Simulation Parameter Settings}
	\label{Tab:AFDMParaSetting}
	\begin{tabular}{|c|c|c|}
		\hline
		Parameter & Symbol & Value \\
		\hline
		Modulation & / & 4-QAM \\
		\hline
		Subcarrier spacing & $\Delta f$ & $15$ kHz\\
		\hline
		Carrier frequency & $f_\mathrm{c}$ & $9$ GHz \\
		\hline
		Number of subcarriers & $N$ & $256$ \\
		\hline
		Maximum Doppler shift & $f_\mathrm{max}$ & $4.17$ kHz \\ 
		\hline
		Maximum iteration number & $\mathrm{Iter}_\mathrm{max}$ & $15$ \\
		\hline
		Normalized maximum delay & $l_\mathrm{max}$ & $6$ \\
		\hline
		Average powers & $\sigma_i^2$ & $\{0,-4.675,-6.482\}$ dB \\
		\hline
		AFDM parameter & $c_2$ & $1/2N$ \\
		\hline
	\end{tabular}
\end{table}

\section{Numerical Results}
\label{sec6}
In this section, numerical results are presented to evaluate the proposed FCDA and HS-JCPS strategies, as well as the performance of the sparsity-exploiting LC-BMMSE receiver. 
We refer to the $3$-tap NTN-TDL-A channel model specified in TR38.811~\cite{3gpp_tr38811}, where the average powers are $\{0,-4.675,-6.482\} \ \mathrm{dB}$.
The maximum normalized  delay is set as $l_\mathrm{max}=6$.
For each tap, the integer delay is randomly generated in the range $\left[0,l_\mathrm{max}\right]$ and the Doppler frequency offset is generated according to the maximum Doppler shift $f_\mathrm{max}=\frac{v}{c}f_\mathrm{c}$, while the Doppler spectrum follows the classical Jakes model over $[-f_\mathrm{max},f_\mathrm{max}]$~\cite{ref16}. 
The remaining system parameters are summarized in Table~\ref{Tab:AFDMParaSetting}.

\begin{figure}[t]
	\centering
	\includegraphics[width = 1\columnwidth]{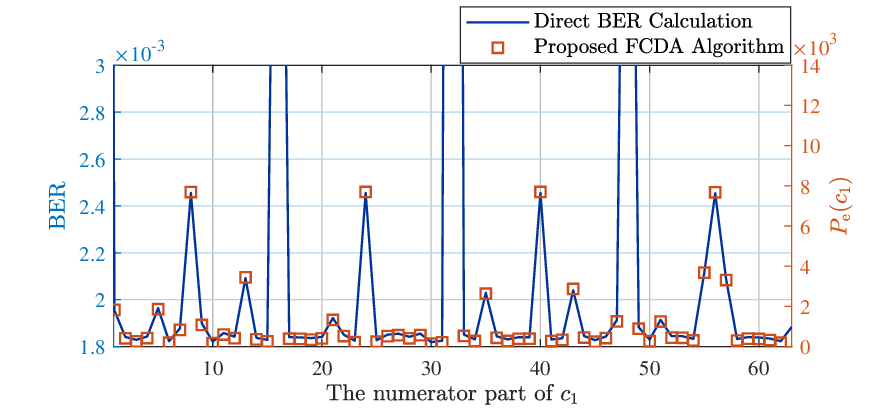}
	\caption{Comparison of BER and the value of $P_\mathrm{e}(c_1)$ over different $c_1$, demonstrating the effectiveness of the FCDA algorithm for identifying the BER-optimal $c_1$.}
	\label{fig:GeneralOptc1_offdiagEnergy}
\end{figure}

\begin{figure}[t]
	\centering
	\includegraphics[width = 1\columnwidth]{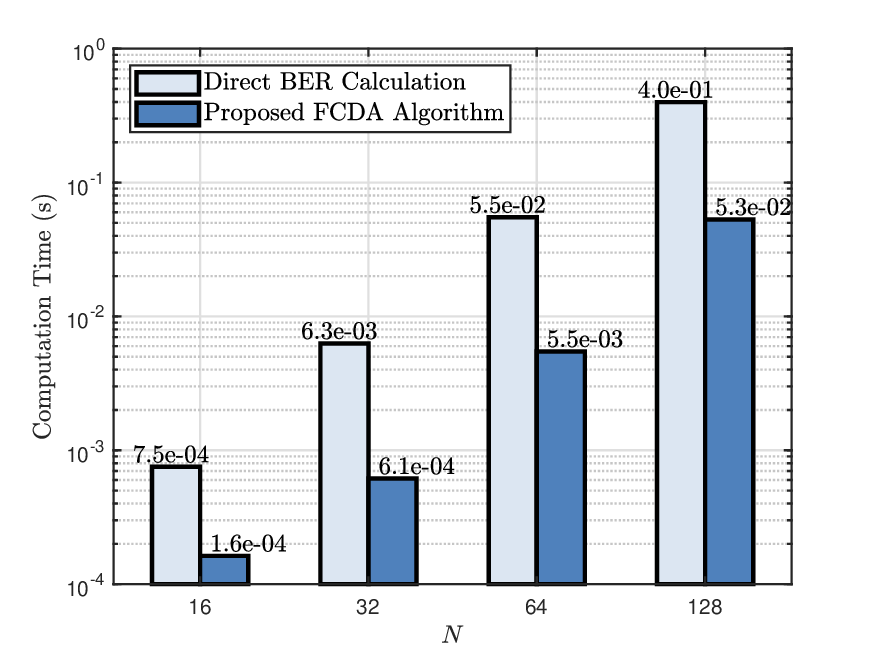}
	\caption{Comparison of computation time between direct BER calculation and the proposed FCDA algorithm for searching all candidate $c_1$ values.}
	\label{fig:Opttime}
\end{figure}

Fig.~\ref{fig:GeneralOptc1_offdiagEnergy} validates the effectiveness of the proposed FCDA algorithm for selecting the AFDM parameter $c_1$. 
{As shown in Fig.~\ref{fig:GeneralOptc1_offdiagEnergy}, $P_\mathrm{e}(c_1)$ follows the variation of the actual BER across different $c_1$ values, with both attaining their minimum at the same $c_1$.
This confirms that the simplified BER metric serves as an
effective and computationally efficient metric for $c_1$ selection.}
Furthermore, Fig.~\ref{fig:Opttime} compares the computational time required by direct BER calculation based on \eqref{eq:Calber} and the proposed FCDA  for optimal $c_1$ selection. 
Both methods were implemented in the same environment under identical settings. 
It can be seen that FCDA reduces the search 
time by approximately one order of magnitude, demonstrating a clear 
computational advantage over direct BER evaluation.
\begin{figure}[t]
	\centering
	\includegraphics[width = 1\columnwidth]{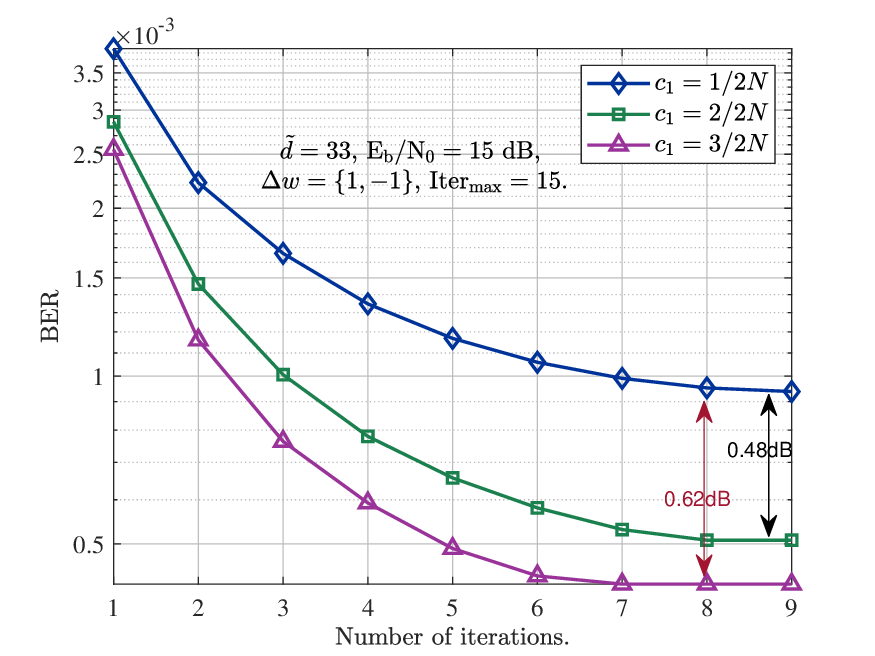}
	\caption{Comparison of BER convergence under different $c_1$ values of  the proposed HS-JCPS strategy.}
	\label{fig:BERiterforHeuAlgo}
\end{figure}

Fig.~\ref{fig:BERiterforHeuAlgo} illustrates the BER convergence behavior 
of the proposed HS-JCPS algorithm under different values of $c_1$, with 
the maximum number of nonzero elements per row and the pruning-width budget 
are set to $\tilde{d} = 33$ and $w_\mathrm{s} = 16$, respectively. Three 
candidate values are considered: $c_1 = 1/2N$, $2/2N$, and $3/2N$. For 
each fixed $c_1$, the inner heuristic search progressively refines the 
pruning-width vector $\mathbf{w}$ through step-wise updates, leading to a lower   BER. In addition,  compared with $c_1 = 1/2N$, the cases 
$c_1 = 2/2N$ and $c_1 = 3/2N$ achieve BER gains of $0.48\,\mathrm{dB}$ 
and $0.62\,\mathrm{dB}$, respectively. Overall, Fig.~\ref{fig:BERiterforHeuAlgo}  shows that     
the proposed HS-JCPS algorithm efficiently explores the discrete search 
space and converges to a near-optimal sparse design within a given 
complexity budget.

\begin{figure}[t]
	\centering
	\includegraphics[width = 1\columnwidth]{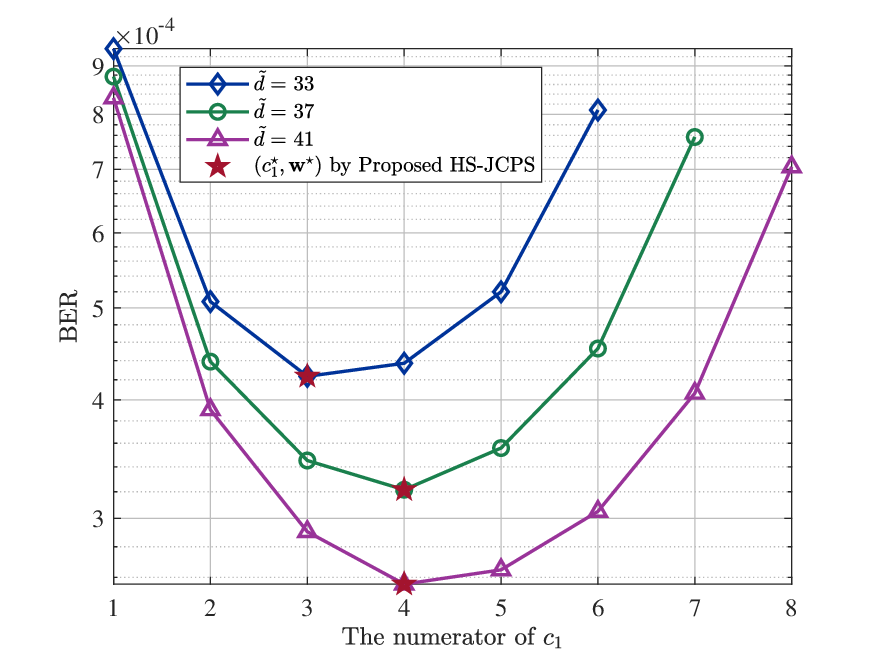}
	\caption{BER comparisons of different candidate pairs $(c_1,\mathbf{w})$ under different maximum numbers of nonzero elements $\tilde{d}$ with $\mathrm{E}_\mathrm{b}/\mathrm{N}_0=15\,\mathrm{dB}$.}
	\label{fig:BERforc1_Diffwidth}
\end{figure}

Fig.~\ref{fig:BERforc1_Diffwidth} compares the BER  of  different candidate pairs $(c_1,\mathbf{w})$ with different values of $\tilde{d}$, i.e., the number of non-zero entries among rows of $\mathbf{S}$ at $\mathrm{E}_\mathrm{b}/\mathrm{N}_0=15\,\mathrm{dB}$. 
For each $\tilde{d}$, the BER does not decrease monotonically with increasing $c_1$, but instead exhibits a first-decreasing-then-increasing trend. 
This is because the receiver only processes the entries of band region located near the main diagonal of sparsified regularized matrix $\tilde{\boldsymbol{\Psi}}$. 
As $c_1$ further increases, the dominant paths may shift significantly beyond the processing range of the receiver, so that their energy can no longer be effectively captured, which in turn leads to performance degradation.
{ The steep slopes around the optimum are caused by the structural sensitivity of the sparsified ECM to $c_1$. Under a fixed truncation width, changing $c_1$ shifts the dominant coupling center through $\xi_i=(\alpha_i+2Nc_1l_i)_N$, so even a small mismatch may move significant path energy outside the retained interval in \eqref{eq:purning}. 
The resulting increase in $\Delta\mathbf H_{\rm eff}$ strengthens the residual interference term $\mathbf D$ in Lemma~\ref{lemma1} and leads to sharp BER degradation.}
This further indicates that MMSE performance is highly sensitive to the 
joint choice of the AFDM chirp parameter and the per-path pruning widths. 
For $\tilde{d} = 33$, the minimum BER is achieved at $c_1 = \frac{3}{2N}$, 
whereas for $\tilde{d} = 37$ the optimum shifts to $c_1 = \frac{4}{2N}$. 
Moreover, although the optimal $c_1$ remains the same for both $\tilde{d} 
= 37$ and $\tilde{d} = 41$, a larger $\tilde{d}$ yields a further BER 
reduction, since a looser complexity budget retains more effective coupling 
terms and provides greater design freedom for joint $(c_1, \mathbf{w})$ 
optimization. These results highlight that the {near-optimal pair} 
$(c_1^\star, \mathbf{w}^\star)$ varies with $\tilde{d}$, underscoring the 
necessity of complexity-aware joint parameter selection.

\begin{figure}[t]
	\centering
	\includegraphics[width = 1\columnwidth]{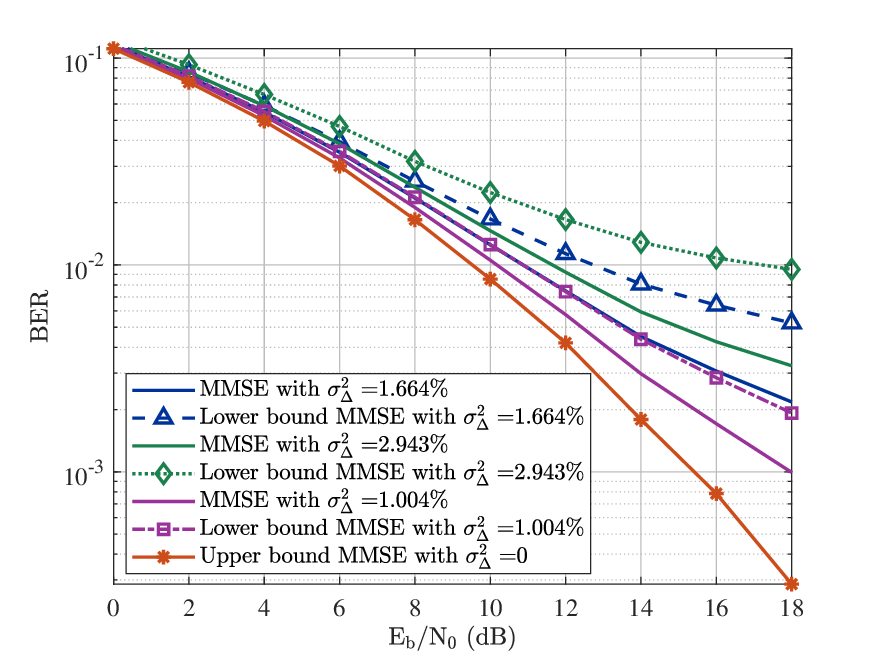}
	\caption{\footnotesize Comparison between MMSE reception and the corresponding upper and lower bounds under different $\sigma_\Delta^2$.}
	\label{fig:EVAforMMSE_LB}
\end{figure}

Fig.~\ref{fig:EVAforMMSE_LB} compares the BER performance of the MMSE detector under different sparsification perturbation levels $\sigma_\Delta^2$. 
The curves labeled ``upper bound'' correspond to full-ECM MMSE detection 
without sparsification, i.e., $\sigma_\Delta^2 = 0$, while the curves labeled ``lower bound''  denote the proposed  bound  given by \eqref{eq:SINRbound}.
As $\sigma_\Delta^2$ increases, the lower bound becomes looser, which is consistent with the fact that a larger sparsification perturbation 
introduces greater mismatch between the full and sparsified effective channel models. 
Nevertheless, minimizing the proposed lower bound remains an effective means of improving the BER performance.

\begin{figure}[t]
	\centering
	\includegraphics[width = 1\columnwidth]{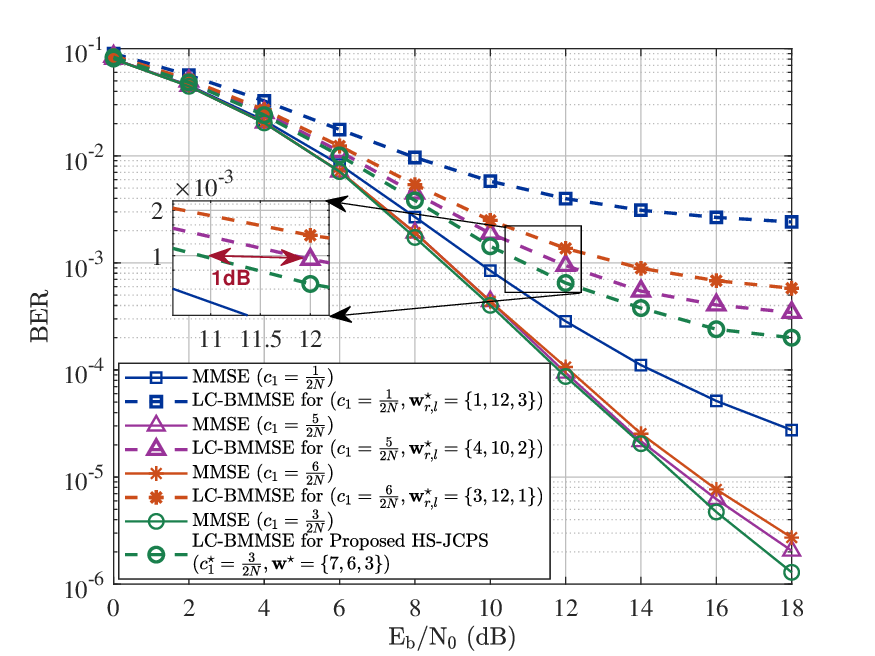}
	\caption{BER performance of the proposed sparsity-exploiting LC-BMMSE receiver versus $\mathrm{E}_\mathrm{b}/\mathrm{N}_0$ for 4-QAM under different sparsification levels of the ECM.}
	\label{fig:BERforSparse}
\end{figure}

Fig.~\ref{fig:BERforSparse} illustrates the BER performance of the proposed HS-JCPS algorithm with the LC-BMMSE receiver. 
Under the pruning-width budget $w_\mathrm{s}=16$ and the DAFT-domain truncation width $d=16$, the BER performance is evaluated versus the noise level. 
The $\mathbf{w}_{r,l}^\star$ values shown in the legend for the LC-BMMSE receiver denote the optimal sparsification patterns identified by the proposed HS-JCPS algorithm for a given candidate $c_1$, while $\mathbf{w}^\star$ denotes {the best sparsification pattern found by the proposed HS-JCPS} over all candidate $c_1$ values.
These results confirm that, under fixed complexity and pruning-width constraints, the proposed HS-JCPS algorithm can identify a near-optimal $c_1$ and pruning width $\mathbf{w}$, achieving an approximately $1$ dB performance gain at a BER of $10^{-3}$.

\begin{figure}[t]
	\centering
	\includegraphics[width = 1\columnwidth]{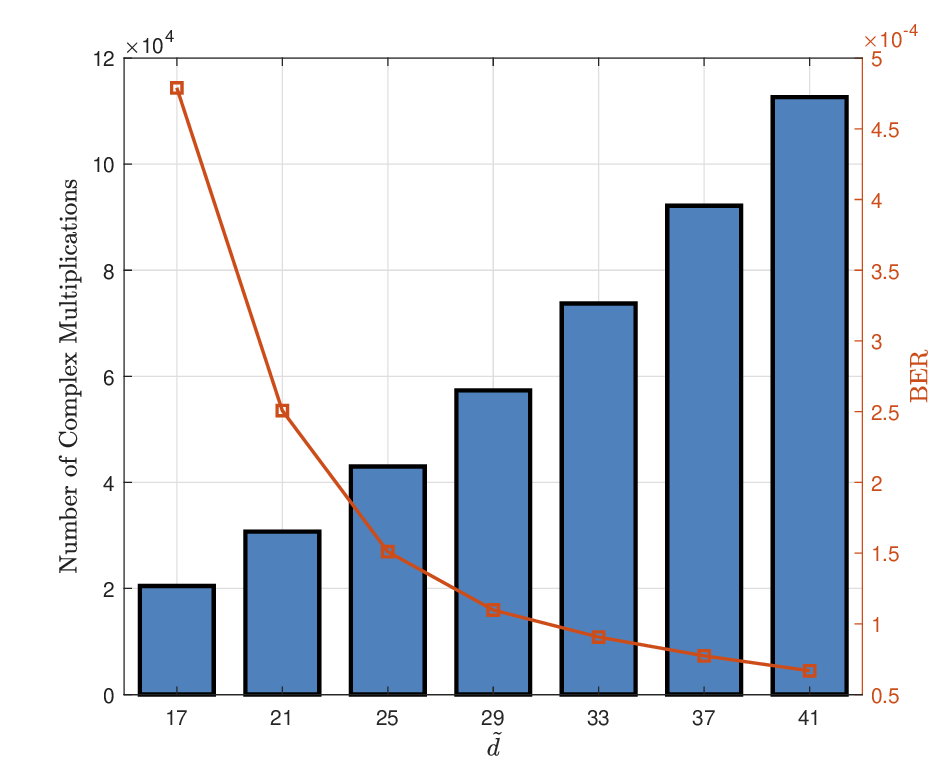}
	\caption{BER-complexity tradeoff under different $\tilde{d}$.}
	\label{fig:BERforComplexwithOptSel}
\end{figure}

Fig.~\ref{fig:BERforComplexwithOptSel} shows the BER-complexity tradeoff under different DAFT-domain truncation width $\tilde{d}$. 
The bar chart illustrates the receiver-side computational cost in terms of the number of complex multiplications, while the curve shows the corresponding BER performance. 
As $\tilde{d}$ increases from $17$ to $41$, the computational complexity rises monotonically from about $2\times10^{4}$ to about $1.1\times10^{5}$ complex multiplications. 
This is because a larger $\tilde{d}$ preserves a wider banded support of the effective model and therefore involves more nonzero coupling terms in detection. 
Moreover, the BER decreases significantly from about $4.8\times10^{-4}$ to about $0.7\times10^{-4}$, indicating that a wider retained neighborhood captures more dominant interference components and improves MMSE detection accuracy. 
However, once $\tilde{d}$ exceeds approximately $33$, the BER improvement becomes marginal whereas the complexity continues to increase. 
Therefore, a moderate $\tilde{d}$ provides a favorable operating point that achieves near-saturated BER performance under a manageable complexity budget.

\begin{figure}[t]
	\centering
	\includegraphics[width = 1\columnwidth]{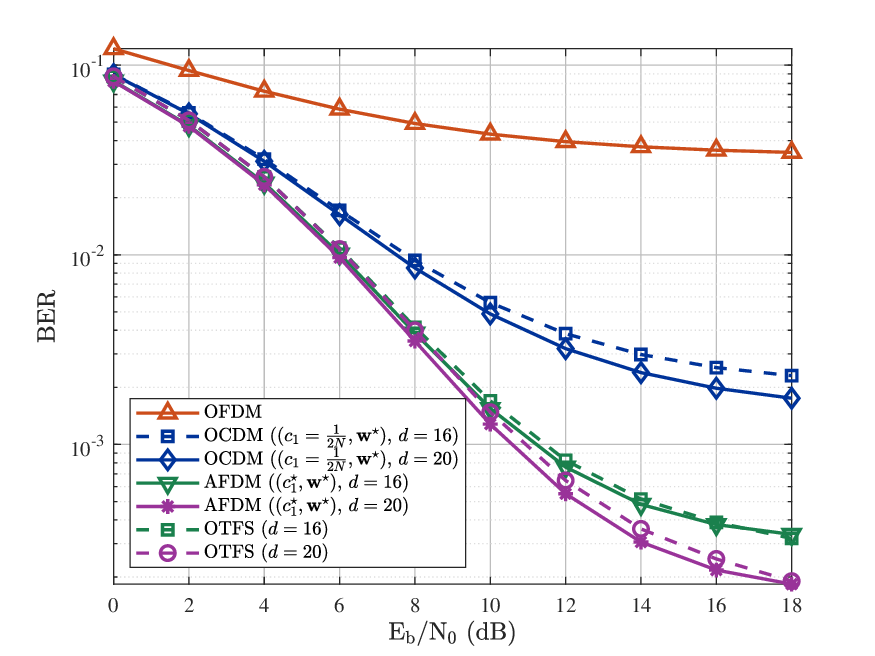}
	\caption{Comparison of BER performance of LC-BMMSE receivers for AFDM, OCDM, OTFS, and OFDM.}
	\label{fig:fig_DiffwaveComp}
\end{figure}

Finally, Fig.~\ref{fig:fig_DiffwaveComp} compares the BER performance of the LC-BMMSE receiver with the proposed HS-JCPS algorithm for different waveforms.
{For a fair comparison, all compared waveforms use the same guard overhead and assume perfect receiver CSI, so pilot overhead and channel estimation errors are not included in Fig.~\ref{fig:fig_DiffwaveComp}.}
Four waveforms are considered, namely, AFDM, OCDM (corresponding  to $c_1 =\frac{1}{2N}$), OFDM (corresponding to $c_1 = 0$) and OTFS, under two complexity constraints with $d=16$  and $d=20$.
The results show that AFDM outperforms both OFDM and OCDM in terms of BER and is comparable to OTFS.
Moreover, owing to its tunable waveform parameter, AFDM offers greater room for performance improvement than OFDM and OCDM, especially when a larger $d$ is allowed.

{To evaluate fractional-Doppler robustness, Fig.~13 shows the BER versus the maximum fractional Doppler shift at $\mathrm{E}_\mathrm{b}/\mathrm{N}_0=15$ dB. 
As the maximum fractional Doppler shift increases, the path energy spreads over wider off-diagonal regions in the DAFT domain, increasing the sparse-ECM truncation error and hence the BER. 
At small the maximum fractional Doppler shift, most channel energy remains within the retained bands, leading to similar performance for $d=16$ and $d=20$. 
Therefore, the configuration with a wider retained support provides greater robustness against fractional Doppler at the cost of increased receiver complexity.}

\begin{figure}[t]
	\centering
	\includegraphics[width = 1\columnwidth]{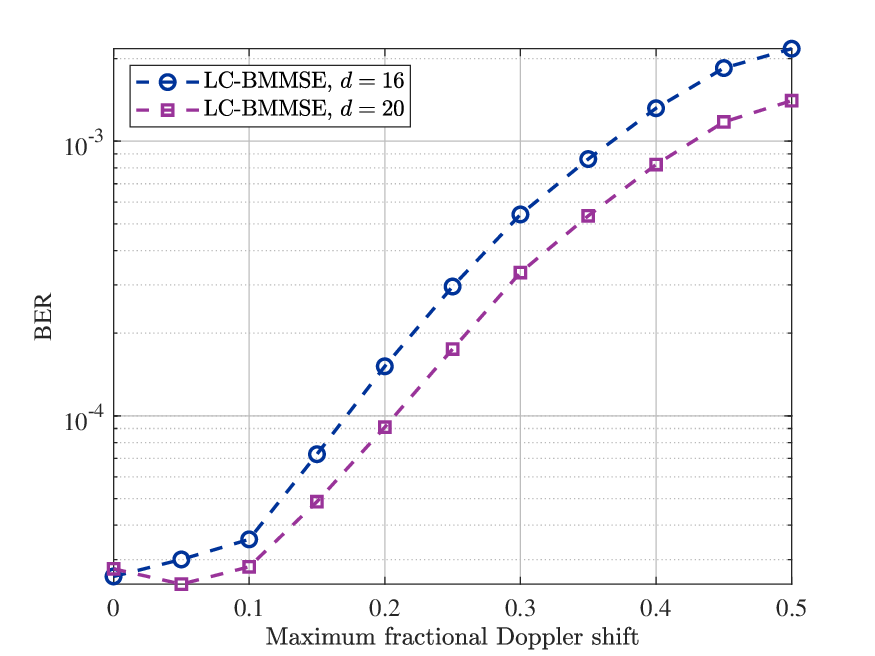}
	\caption{BER performance of the proposed LC-BMMSE receiver versus the maximum fractional Doppler shift under $\mathrm{E}_\mathrm{b}/\mathrm{N}_0=15$ dB.}
	\label{fig:BERforMaxmiumDoppler}
\end{figure}

\section{Conclusion}
\label{sec7}
In this paper, we investigated the joint optimization of chirp-parameter selection and low-complexity MMSE receiver design for AFDM systems. Specifically, we first derived a closed-form BER metric to facilitate performance evaluation under general MMSE detection, based on which an FCDA algorithm was developed for efficient near-optimal chirp-parameter selection. Then, by exploiting the structural characteristics and inherent sparsity of the AFDM ECM, we developed an LC-BMMSE receiver based on reverse Cuthill-McKee (RCM) reordering, cyclic-banded ECM construction, path-wise structured sparsification, and banded Cholesky factorization. {The proposed LC-BMMSE receiver achieves the near-optimal performance with respect to the designed sparse ECM model, while serving as a low-complexity approximation to the full-ECM MMSE receiver.} Finally, to further determine the optimal $c_1$ value under the LC-BMMSE receiver, we derived a lower bound on the MMSE performance associated with the sparsified ECM and proposed an efficient HS-JCPS for joint chirp-parameter and structured-sparsification optimization. Overall, the proposed joint design framework demonstrates that properly integrating chirp-parameter optimization with structure-aware low-complexity MMSE reception can substantially improve the performance-complexity tradeoff of AFDM systems.

\appendices
\section{Derivation of Lower Bound for MMSE estimator}\label{appendixA}

We know that due to the ECM mismatch caused by the sparsification operation, \eqref{eq:SINR_MMSE} is no longer applicable to the accurate calculation of the SINR of the signal processed by the MMSE receiver.
The SINR is calculated using \eqref{eq:SINRk}, resulting in
\begin{equation}
	\mathsmaller{\operatorname{SINR}_k=\frac{\mathrm{E}_\mathrm{s}\left\lvert \left[\tilde{\mathbf{M}}\right]_{k,k}\right\rvert^2}{\mathrm{E}_\mathrm{s}\sum_{j\neq k}\left\lvert \left[\tilde{\mathbf{M}}\right]_{k,j}\right\rvert^2 + \sigma_\mathrm{n}^2\left[\tilde{\boldsymbol{\Psi}}^{-1}\tilde{\mathbf{G}}\tilde{\boldsymbol{\Psi}}^{-1}\right]_{k,k}}.}
	\label{eq:SINR_krew1}
\end{equation}
\eqref{eq:SINR_krew1} shows that the sum of the expected signal power and interference power of the $k$th subcarrier is the $k$th diagonal element of $\tilde{\mathbf{M}}\tilde{\mathbf{M}}^\mathcal{H}$, i.e., $\mathsmaller{\left[\tilde{\mathbf{M}}\tilde{\mathbf{M}}^\mathcal{H}\right]_{k,k}=\left[\tilde{\mathbf{M}}\right]_{k,k}+\sum_{j\neq k}\left\lvert \left[\tilde{\mathbf{M}}\right]_{k,j}\right\rvert^2}$.
$\tilde{\mathbf{M}}\tilde{\mathbf{M}}^\mathcal{H}$ can be derived as
\begin{align}
	\begin{aligned}
		\tilde{\mathbf{M}}\tilde{\mathbf{M}}^\mathcal{H}   =& \ \tilde{\mathbf{M}}\left(\tilde{\boldsymbol{\Psi}}^{-1}\tilde{\mathbf{G}}+\mathbf{D}\right)^\mathcal{H}\\
		= & \left(\tilde{\boldsymbol{\Psi}}^{-1}\tilde{\mathbf{G}}+\mathbf{D}\right) \left(\tilde{\boldsymbol{\Psi}}^{-1}\tilde{\mathbf{G}}\right)^\mathcal{H} +\tilde{\mathbf{M}}\mathbf{D}^\mathcal{H} \\
		= & \left(\tilde{\boldsymbol{\Psi}}^{-1}\tilde{\mathbf{G}}\right)\left(\tilde{\boldsymbol{\Psi}}^{-1}\tilde{\mathbf{G}}\right)^\mathcal{H}+\mathbf{D}\left(\tilde{\boldsymbol{\Psi}}^{-1}\tilde{\mathbf{G}}\right)^\mathcal{H} +\tilde{\mathbf{M}}\mathbf{D}^\mathcal{H}.
		\label{eq:MM}
	\end{aligned}
\end{align}
Multiplying both sides of equation $\tilde{\boldsymbol{\Psi}}^{-1}\tilde{\mathbf{G}}+\frac{\sigma_\mathrm{n}^2+\sigma_\Delta^2}{\mathrm{E}_\mathrm{s}}\tilde{\boldsymbol{\Psi}}^{-1}=\mathbf{I}_N$ by $\tilde{\boldsymbol{\Psi}}^{-1}\tilde{\mathbf{G}}$ and combining this with $\tilde{\mathbf{M}}=\tilde{\boldsymbol{\Psi}}^{-1}\tilde{\mathbf{G}}+\tilde{\boldsymbol{\Psi}}^{-1}\tilde{\mathbf{H}}_\mathrm{eff}^\mathcal{H}\Delta\mathbf{H}_\mathrm{eff}$, we obtain
\begin{equation}
	\tilde{\mathbf{M}} = \tilde{\boldsymbol{\Psi}}^{-1}\tilde{\mathbf{G}} \tilde{\boldsymbol{\Psi}}^{-1}\tilde{\mathbf{G}}+\frac{\sigma_\mathrm{n}^2+\sigma_\Delta^2}{\mathrm{E}_\mathrm{s}}\tilde{\boldsymbol{\Psi}}^{-1}\tilde{\mathbf{G}}\tilde{\boldsymbol{\Psi}}^{-1}+\tilde{\boldsymbol{\Psi}}^{-1}\tilde{\mathbf{H}}_\mathrm{eff}^\mathcal{H}\Delta\mathbf{H}_\mathrm{eff}.
	\label{eq:eqM}
\end{equation}
It is worth noting that the matrix $\tilde{\boldsymbol{\Psi}}^{-1}\tilde{\mathbf{G}}$ is symmetric and satisfies $\tilde{\boldsymbol{\Psi}}^{-1}\tilde{\mathbf{G}}=\left(\tilde{\boldsymbol{\Psi}}^{-1}\tilde{\mathbf{G}}\right)^\mathcal{H}$.
Substituting \eqref{eq:eqM} into \eqref{eq:MM} and focusing only on the diagonal elements, we obtain
\begin{align}
	\begin{aligned}
		&\left[\tilde{\mathbf{M}}\tilde{\mathbf{M}}^\mathcal{H}\right]_{k,k}+\frac{\sigma_\mathrm{n}^2}{\mathrm{E}_\mathrm{s}}\left[\tilde{\boldsymbol{\Psi}}^{-1}\tilde{\mathbf{G}}\tilde{\boldsymbol{\Psi}}^{-1}\right]_{k,k} \\
		&\ =  \left[\tilde{\mathbf{M}}\right]_{k,k} + \left[\tilde{\mathbf{M}}\mathbf{D}^\mathcal{H}\right]_{k,k}-\frac{\sigma_\Delta^2}{\mathrm{E}_\mathrm{s}}\left[\tilde{\boldsymbol{\Psi}}^{-1}\tilde{\mathbf{G}}\tilde{\boldsymbol{\Psi}}^{-1}\right]_{k,k}\\
		&\qquad-\frac{\sigma_\mathrm{n}^2+\sigma_\Delta^2}{\mathrm{E}_\mathrm{s}}\left[\mathbf{D}\left(\tilde{\boldsymbol{\Psi}}^{-1}\right)^\mathcal{H}\right]_{k,k}.
		\label{eq:MMkk}
	\end{aligned}
\end{align}

\begin{proposition}
\label{proposition:1}
Given $\mathrm{E}_{\mathrm{s}}, \sigma_\Delta^2, \sigma_{\mathrm{n}}^2>0$, and $\tilde{\boldsymbol{\Psi}}$ is Hermitian positive definite, the following inequality holds:
\begin{equation}
0<
\frac{\mathsmaller{\sigma_\Delta^2}}{\mathsmaller{\mathrm{E}_\mathrm{s}}}
\left[\tilde{\boldsymbol{\Psi}}^{-1}\tilde{\mathbf{G}}\tilde{\boldsymbol{\Psi}}^{-1}\right]_{k,k}
+
\frac{\mathsmaller{\sigma_\mathrm{n}^2+\sigma_\Delta^2}}{\mathsmaller{\mathrm{E}_\mathrm{s}}}\left[\mathbf{D}\left(\tilde{\boldsymbol{\Psi}}^{-1}\right)^\mathcal{H}\right]_{k,k}<C(\sigma_\Delta),
\end{equation}
where $C(\sigma_\Delta)$ represents the upper bound, which is given by
\begin{equation}
C(\sigma_\Delta)=
\frac{\mathsmaller{\sigma_\Delta^2}}{\mathsmaller{\mathrm{E}_\mathrm{s}}}
\|\tilde{\boldsymbol{\Psi}}^{\mathsmaller{-1}}\|_F^2\|\tilde{\mathbf{G}}\|_F
+
\frac{\mathsmaller{\sqrt{N}\sigma_\Delta(\sigma_\mathrm{n}^2+\sigma_\Delta^2)}}{\mathsmaller{\mathrm{E}_\mathrm{s}}}\|\tilde{\boldsymbol{\Psi}}^{\mathsmaller{-1}}\|_F^2
\|\tilde{\mathbf{H}}_\mathrm{eff}\|_F.
\label{eq:C_fro}
\end{equation}
\end{proposition}
For detailed proof, refer to Appendix~\ref{ProofPro1}.

{
Proposition~\ref{proposition:1} indicates that the residual correction term induced by the SPM is upper-bounded by $C(\sigma_\Delta)$. Since $\Delta\mathbf H_{\rm eff}$ represents the discarded ECM component and enters the sparsified MMSE detector through $\mathbf D$, its average power $\sigma_\Delta^2$ quantifies the model mismatch caused by sparsification. Therefore, a larger $\sigma_\Delta^2$ produces stronger residual interference and a looser SINR lower bound, whereas retaining more ECM energy reduces $\sigma_\Delta^2$ and tightens the bound.
}
Applying Proposition~\ref{proposition:1} to \eqref{eq:MMkk}, we obtain the inequality as
\begin{align}
	\begin{aligned}
		\left[\tilde{\mathbf{M}}\tilde{\mathbf{M}}^\mathcal{H}\right]_{k,k}+\frac{\sigma_\mathrm{n}^2}{\mathrm{E}_\mathrm{s}}\left[\tilde{\boldsymbol{\Psi}}^{-1}\tilde{\mathbf{G}}\tilde{\boldsymbol{\Psi}}^{-1}\right]_{k,k}  \leq  \left[\tilde{\mathbf{M}}\right]_{k,k} + \left[\tilde{\mathbf{M}}\mathbf{D}^\mathcal{H}\right]_{k,k}.
	\end{aligned}
\end{align}

Therefore, the lower bound of SINR for MMSE is expressed as 
\begin{equation}
	\operatorname{SINR}_k \geq \frac{\left[\tilde{\mathbf{M}}\right]_{k,k}}{1+\frac{\left[\tilde{\mathbf{M}}\mathbf{D}^\mathcal{H}\right]_{k,k}}{\left[\tilde{\mathbf{M}}\right]_{k,k}}-\left[\tilde{\mathbf{M}}\right]_{k,k}}.
	\nonumber
\end{equation}

\section{ Proof of \textbf{Proposition}~\ref{proposition:1} }
\label{ProofPro1}
\begin{proof}
The first term $\left[\tilde{\boldsymbol{\Psi}}^{-1}\tilde{\mathbf{G}}\tilde{\boldsymbol{\Psi}}^{-1}\right]_{k,k}\ge 0$ for all $k$, since $\tilde{\mathbf{G}}=\tilde{\mathbf{H}}_\mathrm{eff}^\mathcal{H}\tilde{\mathbf{H}}_\mathrm{eff} \succeq \mathbf{0}$, and the matrix $\tilde{\boldsymbol{\Psi}}$ is Hermitian positive definite.
The second term represents a residual induced by $\Delta\mathbf{H}_\mathrm{eff}$, and is very small when $\lVert \Delta\mathbf{H}_\mathrm{eff}\rVert_F\ll \lVert\tilde{\mathbf{H}}_\mathrm{eff}\rVert_F$.
Hence, the sum is strictly positive.
Moreover, these two items satisfy
\begin{align}
	\begin{aligned}
		\frac{\mathsmaller{\sigma_\Delta^2}}{\mathsmaller{\mathrm{E}_\mathrm{s}}}	\left[\tilde{\boldsymbol{\Psi}}^{-1}\tilde{\mathbf{G}}\tilde{\boldsymbol{\Psi}}^{-1}\right]_{k,k}
	&<\frac{\mathsmaller{\sigma_\Delta^2}}{\mathsmaller{\mathrm{E}_\mathrm{s}}}\|\tilde{\boldsymbol{\Psi}}^{-1}\|_F^2\|\tilde{\mathbf{G}}\|_F,\\
	\frac{\mathsmaller{\sigma_\mathrm{n}^2+\sigma_\Delta^2}}{\mathsmaller{\mathrm{E}_\mathrm{s}}}\left[\mathbf{D}\left(\tilde{\boldsymbol{\Psi}}^{-1}\right)^\mathcal{H}\right]_{k,k} &\leq
	\|\tilde{\boldsymbol{\Psi}}^{-1}\tilde{\mathbf{H}}_\mathrm{eff}^\mathcal{H}
	\Delta\mathbf{H}_\mathrm{eff}
	\Delta\mathbf{I}^\mathcal{H}\|_F\\
	&\le	\frac{\mathsmaller{\sigma_\mathrm{n}^2+\sigma_\Delta^2}}{\mathsmaller{\mathrm{E}_\mathrm{s}}}\|\tilde{\boldsymbol{\Psi}}^{-1}\|_F^2
\|\tilde{\mathbf{H}}_\mathrm{eff}\|_F
\|\Delta\mathbf{H}_\mathrm{eff}\|_F \\
	&= \frac{\mathsmaller{\sigma_\mathrm{n}^2+\sigma_\Delta^2}}{\mathsmaller{\mathrm{E}_\mathrm{s}}}\sqrt{N}\sigma_\Delta\|\tilde{\boldsymbol{\Psi}}^{-1}\|_F^2
\|\tilde{\mathbf{H}}_\mathrm{eff}\|_F.
	\end{aligned}
\end{align}
Combining the above inequalities yields
\begin{equation}
\frac{\sigma_\Delta^2}{\mathrm{E}_\mathrm{s}}
\left[\tilde{\boldsymbol{\Psi}}^{-1}\tilde{\mathbf{G}}\tilde{\boldsymbol{\Psi}}^{-1}\right]_{k,k}
+
\left[\tilde{\boldsymbol{\Psi}}^{-1}\tilde{\mathbf{H}}_\mathrm{eff}^\mathcal{H}
\Delta\mathbf{H}_\mathrm{eff}
\Delta\mathbf{I}^\mathcal{H}\right]_{k,k}
<
C(\sigma_\Delta),
\end{equation}
which completes the proof.
\end{proof}

{\footnotesize
\bibliography{reference.bib}
}
\end{document}